\documentclass[11pt]{article}
\usepackage{amsmath,amsfonts,amsthm,amssymb}
\usepackage{algorithm}
\usepackage{graphics,color,url}
\usepackage{graphicx}
\usepackage{epsfig}
\usepackage{fullpage}
\usepackage{color}

\newtheorem{theorem}{Theorem}[section]
\newtheorem{lemma}{Lemma}[section]
\newtheorem{claim}{Claim}[section]

\newtheorem{example}{Example}[section]
\newtheorem{proposition}{Proposition}[section]

\newcommand{\gpp}{{\sc gpp}}
\newcommand{\gap}{{\sc gap}}
\newcommand{\matroid}{{\sc matroid}}
\newcommand{\matching}{{\sc matching}}
\newcommand{\knapsack}{{\sc knapsack}}
\newcommand{\unit}{{\sc unit}}
\newcommand{\mul}{{\sc mul}}

\newcommand{\pr}{{\mathbf{Pr}}}
\newcommand{\ex}{{\mathbf{E}}}

\newcommand{\calA}{{\cal A}}
\newcommand{\calD}{{\cal D}}

\newcommand{\calL}{{\cal L}}

\title{Mechanism Design without Money via Stable Matching}

\author{Ning Chen\thanks{Division of Mathematical Sciences, School of Physical and Mathematical Sciences, Nanyang Technological University, Singapore. Email: {\tt ningc@ntu.edu.sg, ngravin@pmail.ntu.edu.sg}. } \\
\and
Nick Gravin$^*$\\
\and Pinyan Lu\thanks{Microsoft Research Asia. Email: {\tt pinyanl@microsoft.com}.}}
\date{}

\begin{document}

\maketitle \thispagestyle{empty}

\begin{abstract}
Mechanism design without money has a rich history in social choice literature.
Due to the strong impossibility theorem by Gibbard and Satterthwaite,
exploring domains in which there exist dominant strategy mechanisms is one of the central questions in the field.
We propose a general framework, called the generalized packing problem (\gpp), to study the mechanism design questions without payment.
The \gpp\ possesses a rich structure and comprises a number of well-studied models as special cases, including, e.g., matroid, matching, knapsack, independent set, and the generalized assignment problem.

We adopt the agenda of approximate mechanism design where the objective is to design a truthful (or strategyproof) mechanism without money that can be implemented in polynomial time and yields a good approximation to the socially optimal solution.
We study several special cases of \gpp, and give constant approximation mechanisms for matroid, matching, knapsack, and the generalized assignment problem.
Our result for generalized assignment problem solves an open problem proposed in~\cite{DG10}.

Our main technical contribution is in exploitation of the approaches from stable matching,
which is a fundamental solution concept in the context of matching marketplaces, in application to mechanism design.
Stable matching, while conceptually simple, provides a set of powerful tools to manage and analyze
self-interested behaviors of participating agents.
Our mechanism uses a stable matching algorithm as a critical component and adopts other approaches like random sampling and online mechanisms.
Our work also enriches the stable matching theory with a new knapsack constrained matching model.
\end{abstract}

\newpage

\setlength{\baselineskip}{.45cm}

\section{Introduction}

Algorithmic mechanism design is a fascinating field initiated by the seminal work of Nisan and Ronen~\cite{NR99}, which takes incentive preferences of self-interested agents into account for a multitude of algorithmic challenges. One of the most beautiful designs is that of the VCG mechanism~\cite{vickrey,clarke,groves}, which computes an overall outcome, as well as an individual payment to each participating agent. The VCG mechanism has two remarkable properties: social welfare optimization and incentive compatibility (a.k.a., truthfulness, i.e., it is a dominant strategy for every agent to report his private information to the mechanism).

The story is by no means over, however, due to the very limited applicability of the VCG mechanisms.
First, for a number of realistic problems, the computation of a socially optimal outcome is intractable. Computational efficiency has been one of the key challenges in algorithmic mechanism design. A great amount of work has focused on designing truthful mechanisms that can be implemented efficiently, at a minimum loss of social efficiency, including, e.g., combinatorial auctions~\cite{NR99,DNS06,DD09,DRY11} and machine scheduling~\cite{AT01,AAS05,DDD08}. In other words, these works study the question
``What is the power or limit of computation to approximate a socially optimal solution at a cost of being truthful?".

Second, to enforce a strategic incentive environment, the VCG mechanisms may result in a huge amount of overpayment to the participants~\cite{AT01}. In practice, however, we cannot expect a mechanism to induce  large expenses since the market designer may have his own objectives as well (e.g., revenue maximization). A number of works have therefore focused on designing (computationally efficient) truthful mechanisms that are frugal (i.e., with a small payment)~\cite{AT01,talwar,KKT05,CEG10,KSM10}, or with a sharp budget constraint~\cite{PS10,CGL11}, or even further, that do not have any payment~\cite{DFP08,PT09,LSW10,AFK10,DG10,AFP11}.

There are a number of important domains, e.g., political elections, organ donations, and school admissions, where monetary transfers are strictly prohibited. That is, money cannot be treated as a medium of compensation~\cite{agt-book}. This traces back to the well-studied social choice theory which maps agents preferences to a set of alternatives.
The mechanism design in these settings (a.k.a., social choice function) requires no payment and therefore is more challenging. Specifically, in contrast with the VCG mechanisms, a socially optimal solution no longer suffices for truthfulness.
Hence, a sacrifice in social welfare is a necessity to derive a truthful bidding environment. This raises the following question:
\begin{quote}
{\em What is the power or limit of compensation to approximate a socially optimal solution at a cost of being truthful?}
\end{quote}
We will address the question in the present paper focusing on ``How (is the design)?" and ``How good (is the approximation)?".

Recently, Procaccia and Tennenholtz~\cite{PT09} originated the study of {\em approximate mechanism design without money}, which considers mechanism design questions to achieve truthfulness without using money while providing good approximations to socially optimal solutions. They (and follow up work~\cite{LSW10,AFP10}) considered a facility location problem
where agents report their locations in a metric space and a mechanism chooses some places in the space to install facilities that serve all agents; every agent would like to minimize the distance from his own location to the closest facility. Motivated by kidney exchange, Ashlagi et al.~\cite{AFK10} considered a matching problem where hospitals report individual vertices and a compatible pairing is then established; the objective of every hospital is to have as many of their own matched individuals as possible. Dughmi and Ghosh~\cite{DG10} considered a job scheduling problem in the framework of generalized assignment problem (e.g.,~\cite{ST93,FGM06}) where every job has a preference over machines and reports which machines it is compatible with, and then based on the reported information a mechanism outputs a feasible assignment from jobs to machines.

Notice that the bidding languages in the above examples are all about the {\em combinatorial structure} of the problems. For instance, in the kidney exchange model~\cite{AFK10}, hospitals report the membership of vertices to form an underlying graph; in the job scheduling problem~\cite{DG10}, jobs report compatible edges to form an underlying assignment bipartite graph. Such restrictions on the bidding languages to some extent are necessary in the sense that we have to escape impossibility results of social choice  such as the well-known Gibbard-Satterthwaite theorem~\cite{agt-book}. In particular, as monetary transfers are not allowed, one cannot expect any interesting results if bidding languages are about valuations, even for the simplest single item setting. The work of~\cite{PT09,LSW10,AFK10,DG10,AFP11}, on the other hand, provides truthful mechanisms with good approximations and demonstrates the power of approximate mechanism design without using money.

\subsection{Our Contributions}

We follow the stream of approximate mechanism design without money and propose a natural framework, called the {\em generalized packing problem} (\gpp):
There is a ground set of items whose values are public (e.g., it is common knowledge how much they can be sold for in a marketplace), and every agent holds a subset of items (e.g., the items that the agent is able to produce). The objective is to pick items from all agents so as to maximize the total valuation of selected items given a public known feasibility constraint; this enforces which set of items can be feasibly chosen.

The \gpp\ possesses a rich structure and models a variety of practical settings like resource allocation, job scheduling, and procurement auctions, etc.
It includes a number of well-studied models as special cases in terms of the feasibility constraint, e.g., matroid, knapsack, graph matching (the ground set corresponds to edges), graph independent set (the ground set corresponds to vertices), interval job scheduling, and the generalized assignment problem (the ground set corresponds to job-machine pairs). These problems have a common feature characterized by the \gpp: there are limited resources to satisfy the demands of all of the competitive individuals.

The flexibility of the feasibility constraint enables us to further set restrictions on the number of items that can be picked from each agent. We will consider the following two standard settings in the paper: unit demand (\unit) where at most one item can be selected from each agent, and multi-unit demand (\mul) where there is no restriction on the number of items that can be picked from an agent (as long as the aggregate solution remains feasible). We note that the former is equivalent to saying that on top of the feasibility constraint of the latter, there is an additional partition matroid constraint among all agents, which restricts each one to be unit demand.

While the market designer cares about overall social welfare and would like to search for an optimum (indeed, some instances, like knapsack, job scheduling, and generalized assignment problem, are NP-hard to solve), every agent, on the other hand, is only concerned with his own benefit, which is measured by the total valuation of his own items selected (e.g., how much his items can be sold in the marketplace). Given a mechanism setup by the market designer, every agent reports a subset of items (i.e., the claim from the agent of which items that he is able to produce), and the mechanism then picks item(s) from the claimed subset of each agent so that the aggregate selected items form a feasible solution.

Counter-intuitively, an agent may report a smaller subset to obtain a better outcome for some badly designed protocols (see more discussions in Section~3).
Our objective is therefore to design truthful mechanisms without money that can be implemented in polynomial time and yield good approximations to the optimal solution
with full information. We show different approximation results for different instances, summarized in the table on the next page.
\begin{table*}[t]
\small
\begin{center}
\begin{tabular}{|c|c|c|}\hline
 & Unit demand (\unit) & Multi-unit demand (\mul) \\ \hline
 Matroid & 2 approx$^*$ & optimal$^*$ \\ \hline
 Matching & 3 approx$^*$ & logarithmic approx$^\S$ \\ \hline
 Knapsack (\knapsack) & constant approx$^\dag$ & constant approx$^\S$ \\ \hline
 Gneralized assignment problem (\gap) & constant approx$^\dag$ (main result) & --- \\ \hline
\end{tabular}
\begin{quote}
\hspace{0.2in} $^*$Deterministic mechanisms.
\\ \hspace{0.2in} $^\dag$Universally truthful randomized mechanisms.
\\ \hspace{0.2in} $^\S$Truthful in expectation randomized mechanisms.
\end{quote}
\end{center}
\end{table*}
\normalsize

The work closest to ours is by Dughmi and Ghosh~\cite{DG10} who considered the same model for \gap\ and its variants,
where every agent owns a job and knows the corresponding compatible job-machine pairs (those pairs correspond to items in the language of \gpp).
The authors gave a logarithmic approximation truthful in expectation mechanism, and
left as an open question whether there exists a constant approximation truthful mechanism. We solve the problem with an affirmative answer.
Indeed, our mechanism is universally truthful (i.e., distributed over deterministic truthful mechanisms), which is stronger than truthful in expectation
(i.e., truthful bidding maximizes expected utility).

We note that the matching model considered in our paper is different from the previous work~\cite{AFK10}, where agents hold vertices rather than edges. Further, our model is a generalization of the matching problem considered in~\cite{DG10}: In their setting every agent holds all the edges incident to a vertex that is on one side of a bipartite graph, whereas in our setting the underlying graph is not necessarily bipartite and every agent can keep an arbitrary set of edges. Similarly, the knapsack model considered in our paper allows every agent to keep multiple items, whereas in~\cite{DG10} every agent only holds one item.

All the models considered in the present paper, i.e., matroid, matching, knapsack, and \gap, have different aspects, and one cannot be considered as a special case of another.
Further, our results imply that unit demand and multi-unit demand settings are two rather different models. For instance, solving the multi-unit demand case (e.g., optimally for \matroid-\mul) does not imply the same solution for unit demand. Indeed, we can show that no deterministic truthful mechanism can beat the approximation ratio 2 for \matroid-\unit.
On the other hand, a simple greedy algorithm yields a truthful mechanism for \matching-\unit\ with an approximation ratio of 3, but it is not truthful for \matching-\mul.
A similar difference can be seen between \knapsack-\unit\ and \knapsack-\mul.
Our work therefore broadens the scope of approximate mechanism design without money to a larger domain with more positive results.

\subsection{Techniques}

As discussed above, designing truthful mechanisms without money has an extra constraint and, thus, adds additional difficulties.
Further, the models in our framework \gpp\ are multi-parameter problems, where each agent holds several items;
in the multi-unit demand model, an agent can even win multiple items which can be interconnected through the combinatorial structure of the problem.
Some powerful characterizations, e.g., the monotonicity condition~\cite{AT01}, no longer suffice for truthfulness.
For matroid and matching, due to their desirable combinatorial properties, a simple greedy algorithm still solves the problems nearly optimally (except \matching-\mul).
For knapsack and \gap, however, the greedy can perform arbitrarily badly and will not suffice; new ideas therefore have to be explored.

Our main technical contribution is the application of approaches from the stable matching theory to design truthful mechanisms.
Stable matching, introduced by Gale and Shapley in their seminal work~\cite{GS},
is a fundamental solution concept in the context of matching marketplaces and
has an enormous influence on the design of real world matching protocols. In the setup, there are a set of men and a set of
women, each with a preference ranking over members of the other side; a matching between the
men and women is {\em stable} if there is no man-woman pair who both strictly prefer each other to their current partners.
From the strategic point of view, while Roth's impossibility result (Theorem~4.4,~\cite{roth-book}) implies
that it does not admit any truthful design for both men and women
to claim their preferences truthfully, the men-optimal stable matching
algorithm (i.e., men make proposals in the deferred acceptance
algorithm of Gale and Shapley) is indeed truthful for all men. The same incentive result
holds if the matching is many-to-one (i.e., one side of the market
can have multiple assignments) where the matches are proposed from
the unit demand side of the market (Theorem~5.16,~\cite{roth-book}).

In order to apply the stable matching theory to mechanism design,
we need to define preferences for the two parties, i.e., jobs and machines in \gap.
In particular, since every machine with a knapsack constraint can take multiple jobs, one has to also specify a preference
of the machine over subsets of jobs, not only over individual jobs.
Given carefully designed preferences of jobs and machines,
we run a deferred acceptance stable matching algorithm (jobs make proposals) to generate a feasible assignment,
which is shown to be truthful for a number of special cases.
For the general case of \gap, unfortunately and surprisingly (compared to the incentive result described above for many-to-one matchings), the stable matching algorithm may not guarantee truthfulness.

To derive a truthful mechanism with a constant approximation for the general \gap, we build our mechanism on the stable matching algorithm
as a critical component and use the idea of random sampling: We choose, in expectation, half of the jobs to form a test set $T$ and run the stable matching algorithm on $T$.
The returned stable assignment $\calA^T$ gives a close estimate to the structure of the optimal solution with a high probability.
We then, using the sampled outcome $\calA^T$ as a guidance, compute a real assignment $\calA$ for the rest of the jobs by considering them one by one in a predescribed order.
Stable matching also plays a crucial role in the analysis of the mechanism. In particular,
for every unassigned job-machine pair in a stable assignment, since at least one of them prefers its current assignment, the contribution of either the job or the machine will compensate the loss of the pair. Based on this idea, we establish an upper bound on the optimal solution using a stable assignment $\calA^*$ returned by the stable matching algorithm running on {\em all} jobs. By using a number of desirable properties from stable matchings and probabilistic analysis, we compare the stable assignments $\calA$, $\calA^T$, and $\calA^*$, and show that they are all close to each other, which yields the desired constant approximation.

Finally, it is worth mentioning that the high level structure of our mechanism, i.e., random sampling plus online algorithm fashion, has a similar flavor to the online mechanisms for secretary problems~\cite{BIK,BDG}. However, in the random samplings of those works, the decision for every sampled object is binary, i.e., either pick it or not; in our setting, deciding an assignment for sampled jobs is still a complicated combinatorial problem.
In addition, in the secretary problem objects are assumed to arrive online, whereas ours is a pure offline problem and we take advantage of considering jobs in a fixed order to ensure truthfulness.
This general framework of designing truthful mechanisms (without money) by random sampling plus online mechanisms may find applications in other problems.


\section{Preliminaries}

We consider the following {\em generalized packing
problem} (\gpp): There is a ground set of items
$A=\{a_1,\ldots,a_m\}$, where each item $a_j\in A$ has a public
known valuation $v_j$. In addition, there are $n$ agents where each
agent $i$ holds a disjoint subset of items $A_i\subseteq A$.
The membership of items in $A_i$ is private information and only known
to agent $i$. Our objective is to select subsets of items
$S_i\subseteq A_i$ from all agents with maximum possible total
valuation $\sum_i\sum_{a_j\in S_i}v_j$ given a public known feasibility
constraint, which enforces which sets of items can be feasibly
chosen. The \gpp\ problem includes a number of natural problems as
special cases with respect to different feasibility constraints. We
will consider the following ones in the paper:
\begin{itemize}
\item Matroid. The set $A$ forms a base set of a matroid and the feasible sets are all independent sets of the matroid.\\[-.25in]

\item Matching. Given an underlying graph, the set $A$ corresponds to the set of edges of the graph, and the selected edges have to be a matching in the graph.\\[-.25in]

\item Knapsack. Each item $a_j\in A$ has another parameter capacity $c_j$, and the total capacity of all selected items should be bounded by the knapsack capacity $C$.\\[-.25in]

\item Generalized assignment problem (\gap). There are a number of machines (or bins) and jobs (or goods), where each machine has a capacity and each job has a value and capacity for each of its compatible machines. The ground set $A$ therefore corresponds to the set of compatible job-machine pairs.
\end{itemize}

As discussed in the Introduction, the problem can be considered in the cases where one can
pick either at most one item (\unit) or multiple items (\mul) from each agent.
We will use {\sc problem-unit} and {\sc problem-mul} to denote the two models for the given problem
respectively (e.g., \matroid-\unit\ and \matroid-\mul\ respectively refer to
\matroid\ with unit demand and multi-unit demand).

Agents, as self-interested individuals, have their own objectives as
well and would like to maximize the total valuation of their own item(s) selected. It is
therefore not always the case that an agent reports his true
information $A_i$ to the protocol to his best interest.
Algorithmic mechanism design takes such incentive issues into account with a focus on
managing self-interested behaviors of the agents.
In our framework, upon receiving a submitted bid $B_i$ from each agent, which can be any subset of $A_i$,\footnote{Since the mechanism
selects items from the agents, it is quite risky if an agent reports an item that he is not able to produce (i.e., not in $A_i$).
In particular, what if one such non-existing item is chosen? One possibility is to introduce a penalty (e.g., utility $-1$) for
the agent who cannot provide the item that he claims. All mechanisms considered in the present paper for the unit demand models continue to work
for this variant.}
the mechanism chooses a subset $S_i\subseteq B_i$ from
each agent as an output so that the collection of all $S_i$'s satisfies the feasibility constraint.
The total {\em valuation} of the mechanism output (i.e., social welfare) is $\sum_i\sum_{a_j\in
S_i}v_j$, and the {\em utility} of each agent $i$ is defined to be
$\sum_{a_j\in S_i}v_j$.

For a given protocol, all agents would like to bid strategically to maximize their own utility.
We say a mechanism is {\em truthful} if it is a dominant strategy for every
agent to report his true information $A_i$. That is, for any
submitted bids of other agents, no individual can get a better
outcome by reporting a set that is different from his true information.
If a mechanism is not deterministic, we say the mechanism is
{\em universally truthful} if it takes a random distribution over
deterministic truthful mechanisms, and is {\em truthful in expectation}
if no risk-neutral agent can obtain more utility in expectation by
misreporting his private information.

Note that the mechanisms considered in the current paper only pick
item(s) from each agent as an output, without compensating with
a payment. This is the reason why the classic VCG
mechanism is not applicable even for polynomial time
solvable problems like matroid and matching. At the cost of being truthful and without payment,
we cannot expect the mechanism to work as well as VCG and have an output that maximizes social welfare.
Our focus therefore is to design truthful mechanisms (without money)
that can be implemented in polynomial time and yield good approximations to the optimal social welfare given that all agents report their true information.

\section{Warmup: Matroid and Matching}

We first consider two ``simple" polynomial time solvable
problems matroid and matching, as a warmup, to illustrate the
framework of \gpp\ and approximate mechanism design without money.

Consider the following simple greedy algorithm: Pick items according
to their decreasing order of values $v_1\ge v_2\ge \cdots\ge v_m$
while preserving feasibility. That is, if $S$ is the current selected subset of items from $\{a_1,\ldots,a_{j-1}\}$,
then the algorithm includes item $a_j$ if $S\cup \{a_j\}$ is feasible.
Not surprisingly, the greedy algorithm
yields an optimal truthful mechanism for the \matroid-\mul\
problem, and a 2 approximation truthful mechanism for the
\matroid-\unit\ problem.

For the \matching-\unit\ problem, the greedy algorithm again picks edges according to the decreasing order of their values; once an edge is selected, all edges
incident to its two endpoints or possessed by the same agent will be
removed as they cannot be selected any more. Since one can pick at
most one edge from each agent, the problem corresponds to the intersection of three partition matroids:
one for left-hand side nodes, one for right-hand side nodes, and one for all agents.
The greedy algorithm, therefore, gives a 3 approximation truthful mechanism.

For the \matching-\mul\ problem where multiple edges can be chosen from an agent,
unfortunately, the greedy algorithm does not ensure truthfulness.
We further note that another natural strategy, the maximum matching algorithm, while yielding an optimal solution, also cannot guarantee truthfulness for \matching-\mul, even for \matching-\unit. See examples in Appendix~\ref{appendix-example} for the arguments.
In Appendix~\ref{appendix-mechanism-matching-mul}, we give an $O(\log(m))$ approximation truthful in expectation mechanism for the \matching-\mul\ problem.

On the other hand, for \matroid-\unit, \matching-\unit\ and \matching-\mul, there is a lower bound of 2 on the approximation ratio of any deterministic mechanism. The detailed arguments of this section can be found in Appendix~\ref{appendix-matching}.

\section{Knapsack}

Next we consider the knapsack problem where each item $j$ has a value $v_j$ and a capacity $c_j$, and there is a knapsack constraint $C$ which restricts the subsets of items that can be feasibly chosen. In our model, every agent holds a subset of items (this is why \knapsack\ cannot be considered as a special case of \gap); depending on the demand, either at most one item (\unit) or multiple items (\mul) can be picked from an agent.

While the greedy algorithm --- pick items according to the ratio $\frac{v_j}{c_j}$ --- can do arbitrarily bad,
it is well-known that a variant by complementing the greedy with a single item of the largest value gives a 2 approximation to knapsack.
Can we convert this greedy algorithm to a truthful mechanism in a similar way as what we did for \matroid\ and \matching?
In \matroid\ and \matching, the greedy algorithm processes the items in the order by their values, which is consistent with the preferences of the agents.
For \knapsack, however, the greedy algorithm considers the items in the order of the ratio $\frac{v_j}{c_j}$, while the preferences of the agents are still according to the value $v_j$.
As a result, the greedy algorithm is not truthful.

To derive a truthful mechanism with a good approximation for \knapsack-\unit, we use the idea of random sampling: We choose in expectation half of the agents and compute a solution by the greedy algorithm; we then use it as a guidance to compute a real output by considering the remaining agents one by one. For \knapsack-\mul, there are additional difficulties when we consider each remaining agent: the optimal solution is NP-hard to compute and an algorithm with a good approximation may not guarantee truthfulness.
To get around the issue, we output a subset whose expected total value is equal to that given by the fractional optimal solution (which can be computed easily by the greedy algorithm); this gives truthful in expectation for the considered agent.
All details are referred to Appendix~\ref{appendix-knapsack}.

\section{Generalized Assignment Problem}
\label{sec:GAP}

\newcommand{\gs}{{\sc Gale-Shapley}}

\newcommand{\lsizem}{{\sc gap-mechanism-1}}
\newcommand{\separationm}{{\sc gap-mechanism-2}}
\newcommand{\Gapm}{{\sc gap-mechanism-3}}
\newcommand{\GapMainm}{{\sc gap-mechanism-main}}

\newcommand{\smgreedy}{{\sc sm-greedy}}
\newcommand{\multigs}{{\sc sm-da-alg}}

In this section, we will consider the generalized assignment problem (\gap) where we are
given a set of jobs and a set of machines. Throughout this section,
we will denote jobs by $i$ and machines by $k$. For each job $i$ and machine $k$,
there is a value $v_{i,k}$ and a capacity $c_{i,k}$. Further, each machine $k$ has a capacity constraint $C_k$.
An output of the \gap\ is an {\em assignment} between the jobs and machines where each job is assigned to at most one machine
and the aggregate capacity of the assigned jobs on each machine is within its capacity constraint.

Given a feasible assignment $\calA$ between the jobs and machines, let $\calA_i$ be the machine that job $i$ is assigned to (denote $\calA_i=\emptyset$ if $i$ is not assigned to any machine), and $\calA_{k}=\{i~|~\calA_i=k\}$ be the set of jobs that are assigned to machine $k$. Further, to simplify the notations, let $v(\calA)=\sum_{(i,k)\in \calA}v_{ik}$ be the total value of the assignment, $v(\calA_i)=v_{i,\calA_i}$ be the value of the assignment to job $i$ (denote $v(\emptyset)=0$), and
$v(S)=\sum_{i\in S}v(\calA_i)$ for any subset of jobs $S$. Finally, let
$v(\calA_k)=\sum_{i\in \calA_k}v_{i,k}$ be the total value of the jobs assigned to machine $k$.
The definitions $c(\cdot)$ are defined similarly with respect to capacity.

We follow the strategic consideration studied in~\cite{DG10} where each job $i$ is held by an agent, and the private information that the agent holds is the compatible pairs $(i,k)$ for the machines. Following our \gpp\ language described above, we have a bipartite graph $G$ where one side corresponds to agents/jobs and the other side corresponds to machines, and the edges represent the compatible job-machine pairs. Every edge $(i,k)\in G$ corresponds to an ``item'' whose value $v_{ik}$ and capacity $c_{ik}$ are public. The private information that every agent/job $i$ has is the membership of its compatible edge $\{(i,k)\in G~|~i\}$. All job-machine pairs discussed in this section are with respect to the underlying compatible graph $G$.

For a given mechanism, every job reports a subset of compatible pairs and would like to get assigned to a machine $k$ with maximum possible value $v_{ik}$.
Our objective again is to design a truthful mechanism to approximate the optimal \gap\ solution.
Dughmi and Ghosh~\cite{DG10} gave a logarithmic approximation truthful in expectation mechanism for the problem. In this section, we will show a constant approximation universally truthful mechanism.
Our mechanism is based on stable matching and random sampling. We will first discuss a stable matching algorithm, and then in Section~\ref{section-gap-mechanism} describe the main mechanism.

\subsection{Stable Matching Algorithm}\label{section-stable-matching}

Stable matching is a fundamental solution concept in the context of
two-sided matching marketplaces, introduced by Gale and Shapley in
their seminal work~\cite{GS}. A matching between the two parties of
a bipartite marketplace is called stable if there is no pair who
both strictly prefer each other to their current partners. The
deferred acceptance algorithm of Gale and Shapley~\cite{GS}
computes a stable matching.

In order to apply the stable matching theory to our problem \gap, one needs to define
preferences for the two parties, jobs and machines. For every job
$i$, its preference over machines is pretty straightforward, which
is simply according to the preference of the agent, i.e., ranking
machines by the value $v_{ik}$. For every machine $k$, we may
define its preference over jobs in a similar way according to the
values. For such preferences, the
deferred acceptance algorithm can be implemented in the same way as
the greedy algorithm (i.e., the order of proposals is according to the value $v_{ik}$
and a machine accepts a proposal as long as it has enough remaining capacity).
We denote such implementation by \underline{\smgreedy}.
By the fact that every job goes to at most one machine, it can be seen that \smgreedy\ is truthful for all jobs for any fixed tie breaking rule.

The solution generated by \smgreedy, however, may have an arbitrarily bad total valuation, since
the preferences completely ignore the capacities of the jobs.
To balance the value $v_{ik}$ and capacity $c_{ik}$, we will
consider the following preferences in this section:
\begin{itemize}
\item Every job $i$ has a strict preference list, denoted by $\calL_i$, which is according to the decreasing order of $v_{ik}$.
\item For every machine $k$, its strict preference over individual jobs is according to the decreasing order of the ratio $\frac{v_{ik}}{c_{ik}}$. Since a machine may take multiple jobs, we have to specify the preference of every machine $k$ over subsets of jobs, which is defined as follows: For any two subsets of jobs $S$ and $S'$, where $\sum_{i\in S}c_{ik}\le C_k$ and $\sum_{i\in S'}c_{ik}\le C_k$, first remove all common jobs and then all pairs of jobs in $S$ and $S'$ that have the same ratios $\frac{v_{ik}}{c_{ik}}$, then let $k$ prefer the one which has its most preferred individual job among all the remainings. Finally, $k$ does not prefer any subset whose aggregate capacity is beyond its capacity $C_k$. We denote such strict preference list of every machine by $\calL_k$.\footnote{Alternatively, one can define the preference over feasible subsets of jobs by simply comparing the aggregate values of the jobs in the subsets. Such a preference rule, while capturing the real preference for social welfare, is not implementable as for any given unfeasible set $S$, finding a feasible subset $S'\subset S$ that the machine prefers most is equivalent to solving the knapsack problem, which is known to be NP-hard. Our preference $\calL_k$, however, still balances the values and capacities in terms of their ratios, and is easy to implement in the following stable matching algorithm and the main mechanism.}
\end{itemize}

Note that the preferences $\calL_i$ and $\calL_k$ are {\em strict}. If there are
the same values $v_{ik}=v_{ik'}$ or ratios $\frac{v_{ik}}{c_{ik}}=\frac{v_{i'k}}{c_{i'k}}$,
we always break ties in favor of the smaller capacity;
if the capacities are the same as well, ties are broken in an arbitrary but
fixed order. We say a feasible assignment $\calA$ is {\em stable} if it does not
contain any blocking pair $i$ and $k$, where $i$ strictly prefers
$k$ to $\calA_i$ and $k$ can pick a strictly better subset of jobs in the
collection $\calA_k\cup\{i\}$ (which is easy to verify).


We consider the following deferred acceptance algorithm, called \multigs.

\begin{center}
\small{}\tt{} \fbox{
\parbox{6.0in}{\hspace{0.05in} \\
[-0.05in] \underline{\multigs}
\begin{itemize}
\item Initialize all jobs to be unassigned and set $\calA=\emptyset$.
\item {\bf while} there is an unassigned job which has not proposed to all machines in $\calL_i$
      \begin{enumerate}
      \item \label{item:order}
            Among all pairs $(i,k)$ where $i$ is unassigned and $k$ is his highest ranked unproposed machine,
            pick one with the maximum $\frac{v_{ik}}{c_{ik}}$ and let $i$ propose to $k$.
      \item Let $S=\emptyset$.
      \item For each job $i'\in \calA_k\cup\{i\}$ in the order according to the preference $\calL_k$
            \begin{itemize}
            \item if $S\cup \{i'\}$ is feasible for $k$, let $S\leftarrow S\cup\{i'\}$.
            \end{itemize}
      \item If $i\in S$, set $\calA_k\cup \{i\}\setminus S$ to be unassigned and update assignment $\calA_k\leftarrow S$.
      \item Else (we must have $S=\calA_k$), set $i$ to be unassigned.
      \end{enumerate}
\item Output $\calA$.\\[-.2in]
\end{itemize}
} }
\end{center}

When there are ties in the selection of the largest ratio, i.e.,
$\frac{v_{ik}}{c_{ik}}=\frac{v_{i'k'}}{c_{i'k'}}$, to implement the
algorithm, we use the same tie-breaking rule as described above for $\calL_i$ and $\calL_k$.
Given this tie-breaking rule and the strict preferences of $\calL_i$ and $\calL_k$,
the output of the algorithm is uniquely determined. Since
every job proposes to every machine at most once, the algorithm is guaranteed to terminate.
We further note that Step~1 of the \multigs, picking a pair with the largest ratio, is necessary to
derive the following Theorem~\ref{theorem-gap-4-special-cases}.

For the special cases job value invariant and job capacity invariant
(defined below), Dughmi and Ghosh~\cite{DG10} gave 4 approximation
truthful in expectation mechanisms by an LP-based approach. As we can
see next, the combination of the stable matching algorithm
and greedy algorithm gives a universally truthful mechanism of the
same ratio for these cases, as well two other special cases.

\begin{theorem}\label{theorem-gap-4-special-cases}
Consider the following mechanism: Run either \smgreedy\ or \multigs\
with equal probability. The mechanism is universally truthful with
an approximation ratio of 4 for each of the following cases:
\begin{enumerate}
\item \label{t:i1} (Job value invariant) $v_{ik}=v_{ik'}$ for all machines $k,k'$ and any job $i$. \\[-.25in]
\item \label{t:i2} (Job capacity invariant) $c_{ik}=c_{ik'}$ for all machines $k,k'$ and any job $i$. \\[-.25in]
\item \label{t:i3} (Machine value invariant) $v_{ik}=v_{i'k}$ for all jobs $i,i'$ and any machine $k$. \\[-.25in]
\item \label{t:i4} (Machine capacity invariant) $c_{ik}=c_{i'k}$ for all jobs $i,i'$ and any machine $k$.
\end{enumerate}
\end{theorem}

While the \multigs\ solves the problem nicely for the special cases,
unfortunately (and surprisingly), it is not necessarily truthful in
general; we show an example in Appendix~\ref{example-sm-alg-not-truthful}.

\subsection{Main Mechanism}\label{section-gap-mechanism}

In this section we will use \multigs\ as a critical component to
describe a universally truthful mechanism for the general \gap\
problem. We first give three
truthful mechanisms, then the main mechanism takes a uniform
distribution on these mechanisms. Throughout this section,
$\lambda>2$ is an integer and $\mu>0$ is a constant which are both fed as
parameters to the mechanisms; their values can be taken
appropriately with some conditions described at the end of the
analysis.

We first present two simple deterministic mechanisms, which try to assign those pairs with large values.
The first one is designated to deal with the pairs which
have ``large'' capacities with respect to the capacity of the
corresponding machine.

\begin{center}
\small{}\tt{} \fbox{
\parbox{5.7in}{\hspace{0.05in} \\
[-0.05in] \underline{\lsizem($\lambda$)}
\begin{itemize}
\item Remove all pairs $(i,k)$ with capacity $c_{ik} < \frac{1}{\lambda}C_k$.
\item For the remaining pairs, output an assignment by the \smgreedy\
algorithm with the restriction that each machine can take at most one job.\\[-.2in]
\end{itemize}
} }
\end{center}


The following procedure divides the capacity of each machine evenly into $\lambda$ slots
and assigns at most one job to each slot.

\begin{center}
\small{}\tt{} \fbox{
\parbox{5.7in}{\hspace{0.05in} \\
[-0.05in] \underline{\separationm($\lambda$)}
\begin{itemize}
\item Remove all pairs $(i,k)$ with capacity $c_{ik}>\frac{1}{\lambda}C_k$.
\item For the remaining pairs, output an assignment by the \smgreedy\
algorithm with the restriction that each machine can take at most $\lambda$
jobs.\\[-.2in]
\end{itemize}
} }
\end{center}

Next we will describe the key mechanism, which will apply the stable
matching algorithm \multigs\ as a critical component. In fact, we
will use a slightly different version of \multigs: For every machine
$k$ and a given number $C'_k<C_k$, we define as follows a {\em virtual
capacity} constraint $C'_k$ which provides an additional restriction on the
feasible set: Given a feasible assignment $i_1,i_2,\ldots,i_\ell$ on machine $k$
with preferences $i_1\succ i_2\succ\cdots \succ i_\ell$ on $\calL_k$,
we require $\sum_{j=1}^{\ell-1}c_{i_jk}\le C'_k$
(note that it is allowed that $\sum_{j=1}^{\ell}c_{i_jk} > C'_k$).
That is, while the total capacity of the assigned job on machine $k$
can be larger than $C'_k$, the removal of the least preferred one
must ensure that the total capacity becomes within the virtual
capacity. Running \multigs\ with virtual capacity $C'_k$ means that
in the process of the algorithm, the assigned jobs on machine $k$
always satisfy the virtual capacity constraint.

\begin{center}
\small{}\tt{} \fbox{
\parbox{5.7in}{\hspace{0.05in} \\
[-0.05in] \underline{\Gapm($\lambda,\mu$)}
\begin{enumerate}
\item Remove all pairs $(i,k)$ with capacity $c_{ik}>\frac{1}{\lambda}C_k$.
\item \label{step:2} Select each job independently at random with probability $\frac{1}{2}$ into group $T$.
\begin{itemize}
\item For each job $i\in T$, let $\calL_i$ be his preference list over machines.
\item For each machine $k$, let $\calL_k$ be its preference list over jobs in $T$.
\item Run the \multigs\ with virtual capacity $\frac{\lambda-1}{\lambda}C_k$ for each machine $k$.
\item Denote the generated assignment by $\calA^{T}$.
\end{itemize}
\item For each machine $k$, define the threshold value $t_k=\mu\cdot \frac{v(A_{k}^{T})}{C_k}$.
\item Let $R$ be the remaining jobs that are not selected in $T$, and set $\calA=\emptyset$.
\item For each job $i\in R$ in a given fixed order
\begin{itemize}
\item Let $k=\arg\max_k\left\{v_{i,k}~|~c_{ik}+c(\calA_k)\le C_k\ \textup{and}\ \frac{v_{ik}}{c_{ik}}\ge t_k\right\}$.
\item If $k$ defined above exists, let $\calA_i=k$; otherwise, let $\calA_i=\emptyset$.
\end{itemize}
\item Output $\calA$.\\[-.2in]
\end{enumerate}
} }
\end{center}

In the mechanism, we first sample, in expectation, half of the jobs
in $T$, then on this set $T$ run the stable matching algorithm \multigs. The assignment $\calA^{T}$ from $T$
gives a good estimate to the structure of the optimal solution with a high probability.
We therefore use the average $\frac{v(A_{k}^{T})}{C_k}$, multiplied by a given constant $\mu$, as a
threshold $t_k$ for the jobs that can be assigned to machine $k$. That is,
for the remaining jobs, we explicitly remove those pairs $(i,k)$
with a ratio less than the threshold, i.e., $\frac{v_{ik}}{c_{ik}}
< t_k$. The mechanism then performs in an online and greedy way, trying to have
the best possible assignment for each considered job, given the capacity constraints on the machines.

\begin{claim}
\lsizem($\lambda$) and \separationm($\lambda$) are deterministic truthful mechanisms,
and \Gapm($\lambda,\mu$) is a universally truthful mechanism.
\end{claim}

\begin{claim}\label{cl:stable_3}
The assignment $\calA^T$ returned by the
mechanism \Gapm($\lambda,\mu$) is stable with respect to all jobs in
$T$ and all machines with virtual capacity
$\frac{\lambda-1}{\lambda}C_k$.
\end{claim}

We note that the stability result established in the above claim is only with respect to the virtual capacity constraint;
without which the \multigs\ may not compute a stable assignment.
Stability is helpful in the approximation analysis of the following main mechanism.

\begin{theorem}\label{thm-gap-main-mechanism}
Consider the following mechanism \underline{\GapMainm}: Run one of the following three mechanisms
\lsizem($\lambda$), \separationm($\lambda$), \Gapm($\lambda,\mu$) with uniform probability.
Then \GapMainm\ is a universally truthful mechanism with a constant
approximation ratio for the \gap\ problem.
\end{theorem}

Since all three mechanisms are (universally) truthful,
\GapMainm\ is a universally truthful mechanism as well.
The constant approximation analysis is given in Appendix~\ref{appendix-main-mechanism-approx}.

\section{Concluding Remarks}

The focus of the present paper is on approximate mechanism design without money.
We propose a general framework, the generalized packing problem (\gpp), and
give constant approximation truthful mechanisms for a few special cases of \gpp.
Due to space limit, a number of interesting claims and proofs are deferred to appendix.
As our focus is on exploring domains in which there exist nearly optimal truthful mechanisms,
we do not elaborate on optimizing the approximation ratios.
Closing the gap between upper and lower bounds certainly deserves further study.
Another interesting direction of future work is to further explore other special cases of \gpp\ (e.g., interval job scheduling)
in the framework of approximate mechanism design without money. 

In our model \gpp, we assume that the subsets of items that all agents hold are disjoint.
For \matroid, \matching, and \knapsack, this assumption can be easily removed and all our results still hold.
In particular, if both agents are able to produce the same item, in the output we explicitly specify which agent contributes the item.
Further, it is even possible that both agents contribute the item simultaneously (it means the item will have multiple copies in the output),
as long as it does not violate the feasibility constraint.

Our approaches include greedy algorithms (for all problems), random sampling plus online mechanisms (for \knapsack\ and \gap),
as well as stable matching (for \gap). The fairness condition captured by stable matching helps to
manage self-interested behaviors of the participating agents in one side of a market and to bound efficiency loss for unmatched pairs.
We believe the idea of stable matching may find applications in other mechanism design (without money) problems.

In the solution to \gap, we introduce a new stable matching model where one side of the market has a knapsack constraint,
generalizing the original many-to-one matching market model.
For the new model, the deferred acceptance algorithm, \multigs, may not compute a stable assignment in general,
even if all preferences are strict. This is a significant difference between
the knapsack constrained stable matching model and the original models
where the deferred acceptance algorithm always computes a stable outcome.
It is interesting future work to explore different aspects of stable matchings in the new model, say,
existence, computation, solution structure, and economic properties (e.g., incentives).

\newpage
\appendix

\section{Matroid and Matching}\label{appendix-matching}

\subsection{Greedy Algorithm for Matroid}

\begin{proposition}
The greedy algorithm yields an optimal truthful mechanism for the
\matroid-\mul\ problem, and a 2 approximation truthful mechanism for
the \matroid-\unit\ problem.
\end{proposition}
\begin{proof}
For the \matroid-\mul\ problem, it is well-known that the greedy
algorithm computes an optimal independent set. To prove the
truthfulness, for any fixed bids of other agents, we consider the
bidding behavior of an agent $i$. Assume the order of the items
handled by the greedy algorithm is given by $v_1\ge v_2\ge \cdots\ge
v_n$. Let us denote by $T$ the truthful bids of all agents and by
$T'$ the bids, where agent $i$ has reported a proper subset
$B_i\subset A_i$ instead of $A_i$ and the rest have reported
truthfully. We consider the running process of the greedy algorithm
on the bids $T$ and $T'$; we denote by $T_j$ and $T'_j$,
respectively, the subsets of items that the algorithm selects when
passing the item $a_j$.
For that we observe the following.
\begin{itemize}
\item For any $j$, $|T_j|\ge |T'_j|$. Assume otherwise and consider the
smallest $j$ with $|T_j| < |T'_j|$. Since $B_i$ is a subset of
$A_i$, all items in $T'_j$ are available to the algorithm run for
the bids $T$. There is an item $a_{j'}\in T'_j\setminus T_j$ such
that $T_j\cup \{a_{j'}\}$ is an independent set by definition of the
matroid and since $|T_j| < |T'_j|$ (in case of multiple items, we
pick the item with the smallest index). This implies that
$T_{j'}\cup \{a_{j'}\}$ should be an independent set as well, since
$T_{j'}\subseteq T_j$. Hence, the greedy algorithm should add item
$a_{j'}$ to $T_{j'}$, a contradiction.

\item For any item $a_j\in T_j\setminus T'_j$, it must hold true that
$a_j\in A_i\setminus B_i$ and agent $i$ wins item $a_j$ when bidding
$A_i$. Assume otherwise and consider the smallest $j$, such that
$a_j\in T_j\setminus T'_j$, and either $a_j\notin A_i\setminus B_i$
or another agent wins item $a_j$ for the bids $T$. For either case,
item $a_j$ is also available to the algorithm for bids $T'$
(however, the algorithm cannot pick the item, as the set it would
obtain would not be an independent set). Thus we have
$T'_{j-1}=T'_j$ and $T_j=T_{i-1}\cup\{a_j\}$. Let $\delta =
|T_{j-1}| - |T'_{j-1}|\ge 0$, then we have $\delta+1=|T_j|-|T'_j|$.
Hence, by the definition of a matroid one can find $\delta+1$ items
in $T_j\setminus T'_j$, denoted by $X$, such that $T'_j\cup X$ is an
independent set. Further, $a_j\notin X$, because otherwise, the
algorithm would add $a_j$ into $T'_j$. Then by the choice of $j$, we
have $X\subseteq A_i\setminus B_i$. Now observe that
\[|T'_{j-1}\cup X|=|T'_{j-1}|+\delta+1=|T_{j-1}|+1\]
Hence, there is an item $a_{j'}\in T'_{j-1}\cup X \setminus T_{j-1}$
such that $j'\le j-1$ and $T_{j-1}\cup\{a_{j'}\}$ is an independent
set. Therefore, any subset of $T_{j-1}\cup\{a_{j'}\}$ is an
independent set as well. Moreover, we have $a_{j'}\in
T'_{j-1}\setminus T_{j-1}$, since $X\subseteq T_{j-1}$. Thus, when
the algorithm reaches item $a_{j'}$, it should be added into
$T_{j'}\subseteq T_{j-1}$, a contradiction.
\end{itemize}

Given by the above two observations and the greedy rule, we can see
that the total value of the items that $i$ wins when bidding $A_i$
is always at least that of when bidding $B_i$. This implies that it
is always to the best interest of every agent to report his true
information.

For the \matroid-\unit\ problem, the truthfulness follows trivially
from the facts that one can pick at most one item from every agent;
and that the greedy algorithm always picks the item with the maximum
possible value meanwhile keeping feasibility. The \matroid-\unit\
problem is a two matroids intersection problem, where one matroid is
the original one that defines the feasibility constraint and the
other one is given by the condition that one can pick at most one
item from every agent, which defines a partition matroid over all
items. It is known that the greedy algorithm gives a 2 approximation
for the two-matroid intersection problem, which implies the desired
ratio.
\end{proof}

\subsection{Deterministic Lower Bound}

\begin{proposition}\label{prop-lower-bound}
No deterministic truthful mechanism can have an approximation ratio
better than 2 for \matroid-\unit, \matching-\unit, and
\matching-\mul.
\end{proposition}
\begin{proof}
For the \matching\ problem, consider the graph containing four edges
with values $v(t_1,u_1)=v(t_2,u_1)=1+\epsilon$ and
$v(t_1,u_2)=v(t_2,u_2)=1$, where $\epsilon>0$ is sufficiently small.
There are two agents, where the first agent $i_1$ holds edges
$(t_1,u_1)$ and $(t_1,u_2)$ and the second agent $i_2$ holds edges
$(t_2,u_1)$ and $(t_2,u_2)$. The optimal solution either picks
$\{(t_1,u_1),(t_2,u_2)\}$ or $\{(t_1,u_2),(t_2,u_1)\}$ with a total
value of $2+\epsilon$. If there is a deterministic truthful
mechanism with an approximation ratio better than $2-\epsilon$, then
it has to choose one of the optimal solutions; assume without loss
of generality that it outputs $(t_1,u_1)$ and $(t_2,u_2)$. If $i_2$
hides the edge $(t_2,u_2)$, the mechanism still needs to guarantee
the $2-\epsilon$ approximation for the new instance; hence,
$\{(t_1,u_2),(t_2,u_1)\}$ will be the output. The utility of $i_2$
has therefore increased from $1$ to $1+\epsilon$, which implies that
the mechanism is not truthful.

For the \matroid-\unit\ problem, consider a partition matroid with
ground set $\{a_1,a_2,a_3,a_4\}$. An independent set can take at
most one item from $\{a_1,a_2\}$ and at most one item from
$\{a_3,a_4\}$. The values of the items are $v_1=v_2=1+\epsilon$ and
$v_3=v_4=1$. There are two agents, where the first agent holds
$\{a_1,a_3\}$ and the second agent holds $\{a_2,a_4\}$. The rest of
the analysis is then the same as the above argument.
\end{proof}

We note that the above constructions can be easily generalized to
\knapsack-\unit, \knapsack-\mul, and \gap, with the same lower bound.

\subsection{Examples for \matching}\label{appendix-example}

We claim that the greedy algorithm is not truthful for
\matching-\mul. Consider the example of the following Figure~(B),
where agent $i$ holds three edges $(t_1,u_1),(t_3,u_4),(t_4,u_3)$,
and all other edges are held by separate agents. If all agents
report their edges truthfully, the greedy algorithm will output
$\{(t_1,u_1),(t_2,u_3),(t_3,u_2)\}$, which gives utility 10 to agent
$i$. However, if $i$ hides the edge $(t_1,u_1)$, then
$\{(t_1,u_2),(t_2,u_1),(t_3,u_4),(t_4,u_3)\}$ will be generated by
the greedy algorithm, which gives utility 12 to agent $i$.

\begin{figure}[ht]
\begin{center}
\includegraphics[scale = 0.9]{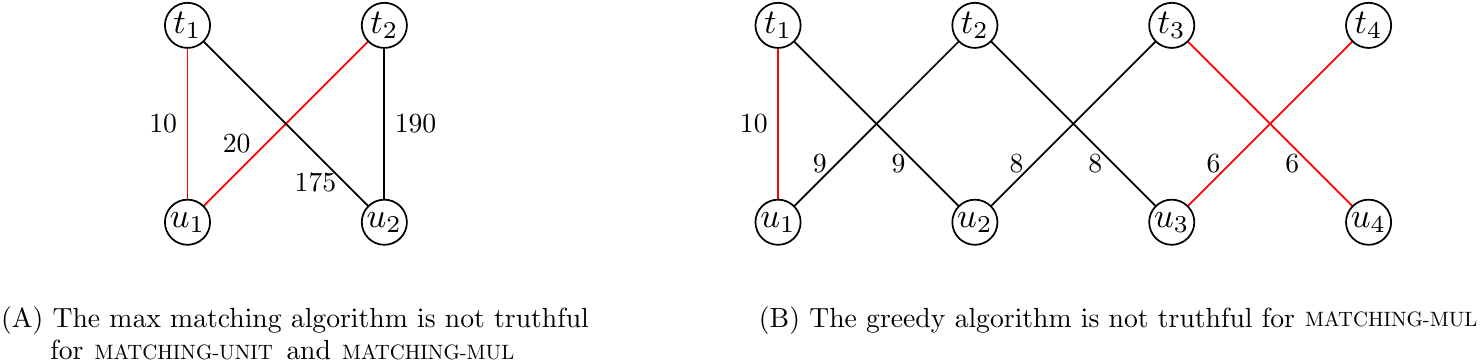}
\end{center}\label{example-matching-greedy}
\end{figure}
\vspace{-0.1in}

In addition, another natural strategy, the maximum matching
algorithm, while yielding an optimal solution, also cannot guarantee
truthfulness for \matching-\mul, even for \matching-\unit. Consider
the example in the above Figure~(A), where agent $i$ holds two edges
$(t_1,u_1)$ and $(t_2,u_1)$. If he bids truthfully, his utility will
be 10 from edge $(t_1,u_1)$; if he hides the edge, then $(t_2,u_1)$
will be matched and his utility will be 20.

\subsection{Mechanism for \matching-\mul}\label{appendix-mechanism-matching-mul}

In this section, we will present a simple $O(\log(m))$ approximation
truthful in expectation mechanism for the \matching-\mul\ problem.
(Recall that $m$ is the number of different items.)

\begin{center}
\small{}\tt{} \fbox{
\parbox{6.3in}{\hspace{0.05in} \\
\underline{\matching-{\sc alg}}: Given submitted bids, let $G=(V,E)$ be the resulting graph
\begin{itemize}
\item Let $v_{\max}=\max_{e\in E}v_e$ and $V$ be the integer s.t. $2^{V-1}<v_{\max}\le 2^V$.
\item For any given fixed order of agents $i_1,\ldots,i_n$, let
$i^*$ be the first one who claims an edge of value $v_{\max}$.
\item Partition all edges according to their values into $\log(m)$ groups:
$$\big(2^{V-\log(m)-1},2^{V-\log(m)}\big],\ldots,\left(2^{V-2},2^{V-1}\right],\left(2^{V-1},2^V\right]$$
(ignore those edges with values less than $2^{V-\log(m)-1}<
v_{\max}/m$).
\item Pick a group uniformly at random and consider only the edges in the picked group.
\item Initially, let $S=\emptyset$, and for each agent $i$ in the order $i^*,i_1,\ldots,i_n$
      \begin{itemize}
      \item Let $S_i$ be a maximum matching restricted to the bid edges of agent $i$ and the set of unmatched vertices by $S$.
      \item Let $S_i$ be the set of edges that agent $i$ wins.
      \item Let $S\leftarrow S\cup S_i$.
      \end{itemize}
\item Output $S$.\\[-.2in]
\end{itemize}
} }
\end{center}

\begin{proposition}
The above \matching-{\sc alg} gives an $O(\log(m))$-approximation
for the \matching-\mul\ problem and is truthful in expectation.
\end{proposition}
\begin{proof}
We first analyze the performance of the algorithm. Note that the
algorithm returns a maximal matching (that is one cannot add any
extra edge to it) for the chosen group. Since the values of all
edges in each group differ by a factor of at most 2, the algorithm
yields a solution that is within a constant factor to the maximum
matching (a matching of the maximal possible weight) of the selected
group. For all the edges with values less than $v_{\max}/m$, their
summation is at most $v_{\max}$, which is at most the optimal
solution of the edges in the group $[v_{\max}/2,v_{\max}]$. Since
the algorithm uniformly chooses a group at random and there are
$\log(m)$ groups in total, the algorithm yields an
$O(\log(m))$-approximation solution to the optimal.

Next we analyze bidding strategies for all agents. The bid of any
agent $i$, which does not have an edge of value $v_{\max}$, does not
affect the way we partition edges into $m$ groups. For any selected
group, the algorithm runs a greedy strategy for each individual
agent in a given fixed order. Hence, it is a dominant strategy for
the agent to report his true information. The same argument extends
to the case when there are at least two agents who claim edges of
the maximum value $v_{\max}$. It remains to consider the case when
$i^*$ is the only agent who claims an edge of the maximum value
$v_{\max}$. May $i^*$ obtain a better outcome by hiding the edge?
Assume that $i^*$ reports another set of edges and $v'_{\max}$ is
the new maximum value for the new instance where
$2^{V'-1}<v'_{\max}\le 2^{V'}$ for some $V'$, then the new partition
of the groups will be

$$\left(2^{V'-\log(m)-1},2^{V'-\log(m)}\right],\ldots,\left(2^{V'-2},2^{V'-1}\right],\left(2^{V'-1},2^{V'}\right]$$

The algorithm will still pick one group uniformly at random and
perform a greedy maximum matching strategy for the selected group.
While agent $i^*$ might be able to obtain more utility from the
groups with smaller values (as they were not available in the
instance with bid $v_{\max}$), the expected benefit is capped by
$\frac{1}{\log(m)}\cdot m \cdot v_{\max}/m=v_{\max}/\log(m)$. On the
other hand, agent $i^*$ loses the chance for the group
$\left(2^{V-1},2^V\right]$ where $v_{\max}$ stays in the new
instance; since $i^*$ has the first priority among all agents, his
utility from the group is at least $v_{\max}$. Hence, $i^*$ also
loses an expected utility of at least $v_{\max}/\log(m)$. For all
other groups that are common in the two instances, the utility that
$i^*$ obtains in the instance with bid $v_{\max}$ is not smaller as
he has the highest priority. Therefore, $i^*$ cannot obtain a higher
expected utility by bidding untruthfully.
\end{proof}

\vspace{-0.2in}
\section{Knapsack}\label{appendix-knapsack}

\newcommand{\sampleksu}{{\sc ks-unit-sample}}
\newcommand{\mechanismksuone}{{\sc ks-unit-mechanism}}

\newcommand{\lagentksm}{{\sc ks-mul-large-agent}}
\newcommand{\sampleksm}{{\sc ks-mul-sample}}
\newcommand{\mechanismksm}{{\sc ks-mul-mechanism}}

\subsection{Mechanism for \knapsack-\unit}

For \knapsack-\unit, we can pick at most one item from the agent
$i$. We denote the aggregate capacity of a set of items $S$ by
$c(S)$ and the aggregate value by $v(S)$. Let $\lambda$ be a sufficiently large
integer, which will be fed as a parameter to our mechanisms. Its value will be determined appropriately at the end of the analysis.
The following is our sample mechanism.

\begin{center}
\small{}\tt{} \fbox{
\parbox{6.0in}{\hspace{0.05in} \\
[-0.05in] \underline{\sampleksu($\lambda$)}
\begin{enumerate}
\item Delete all items with capacity $c_j>\frac{C}{\lambda}$.
\item Pick each agent independently at random with probability $\frac{1}{2}$ into group $T$.
\item Let $R$ be the remaining agents.
\item Run the greedy algorithm for $T$; let $V$ be the total value of the solution.
\item Let $S\leftarrow \emptyset$.
\item For each agent $i\in R$ in a fixed order
\begin{itemize}
  \item Let $X\leftarrow A_i\bigcap \left\{j ~\big|~ \frac{v_j}{c_j}\geq \frac{V}{2C},~ c(S)+c_j \leq C\right\}$.
  \item If $X\neq \emptyset$, let $j^*=\arg\max_{j\in X}v_j$ and $S \leftarrow S \cup \{j^* \}$.
\end{itemize}
\item Output $S$.\\[-.2in]
\end{enumerate}
} }
\end{center}

We first prove a useful lemma, which shows that in the optimal
solution not too many items fail to pass the threshold
$\frac{V}{2C}$.

\begin{claim}\label{lem:half}
Let OPT be the optimal solution for \knapsack-\unit\ problem. Then
$$\sum\limits_{\substack{j\in OPT\\ \frac{v_j}{c_j} \geq \frac{v(OPT)}{2C}}} v_j \geq \frac{1}{2} v(OPT). $$
\end{claim}

\begin{proof}
Assume otherwise, then we have
\begin{eqnarray*}
v(OPT) &=& \sum_{j\in OPT} v_j = \sum_{\substack{j\in OPT\\ \frac{v_j}{c_j} \geq \frac{v(OPT)}{2C}}} v_j +
\sum_{\substack{j\in OPT\\ \frac{v_j}{c_j} < \frac{v(OPT)}{2C}}} v_j \\
&<&  \frac{1}{2} v(OPT) + \frac{v(OPT)}{2C} \cdot \sum_{\substack{j\in OPT\\ \frac{v_j}{c_j} < \frac{v(OPT)}{2C}}} c_i \\
&\leq& \frac{1}{2} v(OPT) + \frac{1}{2}v(OPT) = v(OPT),
\end{eqnarray*}
a contradiction.
\end{proof}

Note that the value $V$ computed in the mechanism is at most
$v(OPT)$. Therefore,
$$\sum_{\substack{j\in OPT\\ \frac{v_j}{c_j} \geq \frac{V}{2C}}} v_j \ge \sum_{\substack{j\in OPT\\ \frac{v_j}{c_j} \geq \frac{v(OPT)}{2C}}} v_j \geq \frac{1}{2} v(OPT). $$

The following is our mechanism for \knapsack-\unit.

\begin{center}
\small{}\tt{} \fbox{
\parbox{5.0in}{\hspace{0.05in} \\
[-0.05in] \underline{\mechanismksuone}
\begin{enumerate}
\item With probability $\frac{1}{2}$, output an item with the largest value.
\item With probability $\frac{1}{2}$, run \sampleksu($\lambda$).
\end{enumerate}
} }
\end{center}

\begin{theorem}\label{theorem-ks-unit}
\mechanismksuone\ is universally truthful and has a constant
approximation ratio.
\end{theorem}

\begin{proof}
The truthfulness for the first part where we select a single item
with the largest value is straightforward. For
\sampleksu($\lambda$), agents in $T$ do not have incentive to lie as
none of their items can be selected. For each agent $i\in R$, since we
try to select the best possible item from him given the capacity and
threshold constraints, $i$ does not have an incentive to lie. Thus,
both \sampleksu($\lambda$) and \mechanismksuone\ are universally
truthful.

Next we analyze the approximation ratio of \mechanismksuone.
Assume without loss of generality that the optimal value is $1$,
since we can scale simultaneously all the values without affecting
the outcome. If there is an item $j$ with value $v_j\ge
\frac{1}{2\lambda}$, then with probability $\frac{1}{2}$, we will
pick an item with the highest value, which is at least $v_j$. This
already implies a constant approximation. Therefore, we assume that
the value of each item is less than $\frac{1}{2\lambda}$.  Now if we
delete all items with size larger than $\frac{C}{\lambda}$, in
the optimal solution we will delete at most $\lambda$ items, whose
total value is at most $\frac{1}{2}$. Therefore, the total value of
the remaining items will be at least $\frac{1}{2}$. Thus, it
suffices to prove that \sampleksu($\lambda$) achieves a constant
approximation with respect to the optimal solution on the setting
where we have already removed all items of large size. In
conclusion, we may restrict ourself to the case where each item has
a small value $v_j<\frac{1}{2\lambda}$ and small size $c_j\le
\frac{C}{\lambda}$.

Let $${\rm OPT}^s = \left \{j ~\Big|~ j\in {\rm OPT}, c_j \leq
\frac{C}{\lambda}\right \}$$ and $${\rm OPT}^* = \left \{j ~\Big|~
j\in {\rm OPT}^s, \ \frac{v_j}{c_j} \geq \frac{v({\rm
OPT}^s)}{2C}\right \}.$$ By the above argument and
Claim~\ref{lem:half}, we know that $v({\rm OPT}^s) \geq \frac{1}{2}$
and   $v({\rm OPT}^*) \geq \frac{v({\rm OPT}^s)}{2} \geq
\frac{1}{4}$.

In the algorithm, we put each agent either in group $T$ or $R$
independently at random with probability $\frac{1}{2}$. Note that
every item in ${\rm OPT}^s$ belongs to a different agent.
Since all items have a small value, we can apply Lemma
\ref{lem:probability} in Appendix~\ref{appendix-main-mechanism-approx}, which tells us that with a high probability,
$T$ will get ``roughly'' half of the total value of ${\rm OPT}^s$ and
$R$ will get ``roughly'' half of the value of ${\rm OPT}^*$. By
choosing a sufficiently large $\lambda$ (according to Lemma
\ref{lem:probability}, $\lambda\ge 2\cdot 36$), with
probability at least $0.75$, we get $v(T\cap {\rm OPT}^s) \geq
\frac{1}{3}v({\rm OPT}^s)\geq \frac{1}{6}$. Similarly, with
probability at least $0.75$, for sufficiently large $\lambda$
($\ge 4\cdot 36$) we get $v(R\cap {\rm OPT}^*)
\geq \frac{v({\rm OPT}^*)}{3} \geq \frac{1}{12}$. Thus with
probability at least $0.5$, we have
$$ v(T\cap {\rm OPT}^s) \geq \frac{v({\rm OPT}^s)}{3} \geq \frac{1}{6} \ \ \  \mbox{ and }\ \ \  v(R\cap {\rm OPT}^*) \geq \frac{v({\rm OPT}^*)}{3} \geq \frac{1}{12}.$$
If the last two inequalities both hold, then the value $V$ in
\sampleksu($\lambda$) is at least $\frac{1}{12}$, as all items have
a small size and therefore the greedy algorithm provides an
approximation ratio smaller than $2$.

There are two possibilities regarding the process of the mechanism.
\begin{itemize}
\item \sampleksu($\lambda$) at some step has taken nothing from an agent $i$, but $i$ has an
item with $\frac{v_j}{c_j}\geq \frac{V}{2C}$. In this case $S$
should occupy at least $(1-\frac{1}{\lambda})C$ capacity of the
knapsack, since the size of the item is less than or equal to
$\frac{C}{\lambda}$. Since all chosen items have
$\frac{v_j}{c_j}\geq \frac{V}{2C}$, the total value of the mechanism
is at least $(1-\frac{1}{\lambda}) \frac{V}{2}\ge \frac{1}{48}$
(given that $\lambda\geq 2$).

\item \sampleksu($\lambda$) has taken an item from every agent that has at
least one item with $\frac{v_j}{c_j}\geq \frac{V}{2C}$. Note that
every item in $R\cap {\rm OPT}^*$ satisfies this condition. So for
each item in $R\cap {\rm OPT}^*$, either this or another one but
with larger value is chosen from the same agent. Therefore, the
total value of items chosen by the mechanism is at least $v(R\cap
{\rm OPT}^*)\ge \frac{1}{12}$.
\end{itemize}

For both cases, the expected value of the mechanism is at least a
constant, which provides a constant approximation to the optimal
value 1. This completes the proof.
\end{proof}

\subsection{Mechanism for \knapsack-\mul}

In \knapsack-\mul, we may choose several items from an agent. The
most natural idea one may consider here is to extend the mechanism for
\knapsack-\unit\ to the multi-demand setting. However, there are
several difficulties. First, in the \sampleksu\ algorithm for
\knapsack-\unit, in order to use the probabilistic lemma, we need to
bound the maximum value of a single item. This is sufficient for
\knapsack-\unit\ since we choose at most one item from an agent. For
\knapsack-\mul, however, we need to bound the total contribution of
a single agent in order to use the lemma. Therefore, we need to
employ one more mechanism to handle such single agent. But how can
we choose the agent? This raises the second difficulty: The
computation of the optimal solution for a single agent is NP-hard
(it is equivalent to solving the knapsack problem). On the other
hand, if we use the greedy algorithm, the mechanism may not be
truthful.

To overcome these difficulties, we use the fractional optimal
solution which can be computed by the greedy algorithm plus a
fraction of the first item that the greedy does not pick so that the
aggregate capacity is equal to the full capacity. We first ignore
all items with capacity larger than $\frac{C}{2}$. Then for each agent we compute the fractional optimal
solution with total capacity smaller than or equal to $\frac{C}{2}$ and choose the
best solution over all agents. There is one nice property of the fractional
optimal solution provided by the greedy rule: there is at most one
item which is fractionally picked. We will include this item with
a probability equal to its share in the fractional optimal
solution. The mechanism is given below.

\begin{center}
\small{}\tt{} \fbox{
\parbox{5.8in}{\hspace{0.05in} \\
[-0.05in] \underline{\lagentksm}
\begin{itemize}
\item Remove all items with capacity $c_j>\frac{C}{2}$.
\item For each agent $i$, compute the fractional optimal solution $F_i$ with total capacity at most
$\frac{C}{2}$.
\item Choose an agent $i^*$ with the largest value of $F_i$.
\item Assume the fractional solution $F_{i^*}$ is composed of
      \begin{itemize}
      \item a subset of items $X$ that are fully picked,
      \item an item $j$ that is picked with a share $0< \alpha<1$.
      \end{itemize}
\item Output
      \begin{itemize}
      \item $X$ at probability $1-\alpha$;
      \item $X\cup \{j\}$ at probability $\alpha$.
      \end{itemize}
\end{itemize}
} }
\end{center}

We note that since we reserve half of the total capacity and the
size of each item is not more than $\frac{C}{2}$, including the
last item $j$ still gives a feasible solution. In addition, by
making the probability equal to the share of the item $j$,
the agent's expected utility is the same as his fractional optimal value.
This ensures the truthfulness (in expectation) for the agent.

\begin{claim}\label{lem:lagentksm-truthful}
\lagentksm\ is a truthful in expectation mechanism.
\end{claim}
\begin{proof}
No agent can enlarge his fractional optimal solution by reporting
smaller set. The agent $i^*$ does not have any incentive to lie
since his expected revenue is exactly his value for $F_{i^*}$ and
she may not increase this value by lying. Any other agent cannot get
anything, since she only is allowed to decrease the size of reported
set. Therefore, mechanism is  truthful-in-expectation.
\end{proof}

There is a similar difficulty in extending \sampleksu\ to the
multi-unit demand setting: The greedy algorithm may not be truthful
for the agents in $R$. Again the similar idea of fractional
optimal solution helps us to get round this difficulty.


\begin{center}
\small{}\tt{} \fbox{
\parbox{6.2in}{\hspace{0.05in} \\
[-0.05in] \underline{\sampleksm($\lambda$)}
\begin{enumerate}
\item Delete all items with capacity $c_j>\frac{C}{\lambda}$.
\item Pick each agent with probability $\frac{1}{2}$ into group $T$; let $R$ be the remaining agents.
\item Run the greedy algorithm for $T$; let $V$ be the total value of the solution.
\item For each agent in $R$, remove all his items with $\frac{v_j}{c_j}<\frac{V}{2C}$.
\item Let $S\leftarrow\emptyset$ and $\widetilde{C}=C-\frac{C}{\lambda}.$
\item For each agent $i\in R$ on a fixed order
  \begin{enumerate}
  \item Compute the fractional optimal solution $F_i$ for the items of agent $i$ within the capacity $\widetilde{C}-c(S).$
  \item Let $F_{i}=X\cup \{j\}$ where items in $X$ are fully picked and $j$ is only picked\\ with a share $0< \alpha<1$.
  \item With probability $\alpha$, let $S\leftarrow S\cup X\cup \{j\}$; with probability $1-\alpha$, let $S\leftarrow
        S\cup X$;\\ if $\alpha>0$ stop.
  \end{enumerate}
\item Output $S$.
\end{enumerate}
} }
\end{center}

Notice that the final outcome $S$ with certain probability may
exceed the capacity $\widetilde{C}$, but it is always within the
capacity $C$ (this follows from the fact that we set
$\widetilde{C}=C-\frac{C}{\lambda}$ at the beginning).

Our mechanism for \knapsack-\mul\ is as follows.

\begin{center}
\small{}\tt{} \fbox{
\parbox{5.2in}{\hspace{0.05in} \\
[-0.05in] \underline{\mechanismksm}
\begin{enumerate}
\item With probability $\frac{1}{3}$, choose a single item with the largest value.
\item With probability $\frac{1}{3}$, run \lagentksm.
\item With probability $\frac{1}{3}$, run \sampleksm($\lambda$).
\end{enumerate}
} }
\end{center}

\begin{theorem}
\mechanismksm\ is truthful in expectation and has a constant approximation ratio for \knapsack-\mul.
\end{theorem}

\begin{proof}
The truthfulness for the first part where we choose a single item
with the largest value is straightforward. Truthfulness for
\lagentksm\ has been shown in Claim~\ref{lem:lagentksm-truthful}.
The similar argument works for \sampleksm($\lambda$) as well.

By scaling all values simultaneously, without loss of
generality we may assume that the optimal value is $1$. If there exists a single
item with value at least $\frac{1}{4\lambda}$, then we are done,
since with probability at least $\frac{1}{3}$ we will pick either
this item or better in the output.

In the following we assume that the value of any single item is at
most $\frac{1}{4\lambda}$. Let OPT be the optimal solution and
$v(OPT)$ be its value. Consider the possible contribution of a
single agent to OPT. If there exists an agent who contributes at
least $\frac{1}{2\lambda}$ to OPT, then after removing of all his items with
capacity larger than $\frac{C}{2}$, at most one of his items from OPT has been removed.
Therefore, his contribution to OPT will be still at least
$\frac{1}{2\lambda}-\frac{1}{4\lambda}=\frac{1}{4\lambda}$. Then the
value returned by \lagentksm\ is at least $\frac{1}{8\lambda}$. We
are also done.

It remains to consider the case where the contribution of any single
agent to OPT is at most $\frac{1}{2\lambda}$. We next will prove
that \sampleksm($\lambda$) contributes a constant value. By removing
all the items with capacity larger than $\frac{C}{\lambda}$, we could
remove at most $\lambda$ items from OPT and the value of remaining
items in OPT is at least $1-\lambda\cdot \frac{1}{4\lambda}=\frac{3}{4}$.
Hence, it is sufficient to prove that \sampleksm($\lambda$) has a
constant approximation ratio providing that every item has size less
than or equal to $\frac{C}{\lambda}$ and every agent contributes
less than $\frac{1}{2\lambda}$ to the optimal solution.

The rest of the proof is similar to that of Theorem~\ref{theorem-ks-unit}, and thus, is omitted here. 
We conclude that \mechanismksm\ has a constant approximation. 
\end{proof}

\section{Generalized Assignment Problem}

\subsection{Proof of Theorem~\ref{theorem-gap-4-special-cases}}

\begin{proof}
Since \smgreedy\ is truthful, it suffices to show that \multigs\ is
truthful as well. For all four invariant settings, the process of
the \multigs\ has the following nice property: Once $i$ and $k$ are
matched, the assignment will never be broken. We will prove this
only for the job value invariant setting; the arguments for the job
capacity invariant and machine capacity invariant settings are
similar. We consider the first pair picked by the algorithm, $i$ and
$k$. According to the algorithm, the ratio given by the first
assigned pair, $\frac{v_{ik}}{c_{ik}}$, is the largest possible
among all ratios. Since every job has the same value $v_i$ for all
machines, all jobs are essentially indifferent between all machines.
Thus machine $k$ has the highest rank in $\calL_i$. Therefore, after
job $i$ is assigned to machine $k$, the assignment will never be
broken. We can then apply the same argument iteratively for all
assigned pairs. (The argument for the machine value invariant
setting is slightly different: Since all jobs have the same value
for every machine, all of them have the same preference over
machines. Hence, the algorithm runs in such a way that all jobs try
to fill the machine that they prefer most, according to the
increasing order of their capacities, then move on to the next
machine, and so on.) Given the above property, we know that the
assignment returned by the algorithm is stable, and the algorithm
runs in a greedy manner. Therefore, it is truthful for all jobs.

Next we will prove that the approximation ratio of the mechanism is
4. We denote by $\calA^{(1)}$, $\calA^{(2)}$ the assignments we get
from \smgreedy\ and \multigs, respectively, and by $\calA^{opt}$ the
optimal assignment.

The expected value of the mechanism is given by
$\frac{1}{2}v(\calA^{(1)})+\frac{1}{2}v(\calA^{(2)})$. We want to
show that
\[4\cdot \left(\frac{1}{2}v\big(\calA^{(1)}\big)+\frac{1}{2}v\big(\calA^{(2)}\big)\right) = 2v\big(\calA^{(1)}\big)+2v\big(\calA^{(2)}\big) \ge v\big(\calA^{opt}\big).\]

We denote by $S$ the set of jobs $i$ for which
$v\big(\calA_{i}^{opt}\big)> v\big(\calA_{i}^{(1)}\big)$ and
$v\big(\calA_{i}^{opt}\big)> v\big(\calA_{i}^{(2)}\big)$. Then, we
get
\[
\sum\limits_i v\big(\calA_{i}^{(1)}\big)+\sum\limits_i v\big(\calA_{i}^{(2)}\big)\ge \sum\limits_{i\notin S} v\big(\calA_{i}^{opt}\big).
\]
Further, we have $v(\calA^{opt})=\sum\limits_{i\notin S} v\big(\calA_{i}^{opt}\big)+\sum\limits_{i\in S} v\big(\calA_{i}^{opt}\big) = \sum\limits_{i\notin S} v\big(\calA_{i}^{opt}\big)+\sum\limits_{k}v\left(\calA_{k}^{opt}\cap S\right).$

Now consider any job $i_0\in S$; let $\calA_{i_0}^{opt}=k$. By the
definition of $S$, we know that in both assignments $\calA^{(1)}$
and $\calA^{(2)}$, job $i_0$ prefers machine $k$ to its actual
assignments. Therefore, for $\calA^{(1)}$, by the greedy rule, we
have $v\big(\calA_{k}^{(1)}\big)\ge v_{ik}=
v\big(\calA_{i}^{opt}\big)$ for any $i\in \calA_{k}^{opt}\cap S$.

For $\calA^{(2)}$, by the property of stability, for any $i\in
\calA_{k}^{opt}\cap S$, we have (i) $v\big(\calA_{i}^{opt}\big)\le
c\big(\calA_{i}^{opt}\big)\frac{v\big(\calA_{i'}^{(2)}\big)}{c\big(\calA_{i'}^{(2)}\big)}$
for each $i'\in\calA_{k}^{(2)}$, and (ii)
$C_k<c\big(\calA_{k}^{(2)}\big)+c\big(\calA_{i}^{opt}\big)$, since
we cannot add $i$ to $\calA_{k}^{(2)}.$ Then we get
\[
v\left(\calA_{i}^{opt}\right)\le c\left(\calA_{i}^{opt}\right)\frac{\sum\limits_{i\in\calA_{k}^{(2)}}v\left(\calA_{i}^{(2)}\right)}
{\sum\limits_{i\in\calA_{k}^{(2)}}c\left(\calA_{i}^{(2)}\right)}=c\left(\calA_{i}^{opt}\right)\frac{v\left(\calA_{k}^{(2)}\right)}{c\left(\calA_{k}^{(2)}\right)}.
\]
For the considered job $i_0\in\calA_{k}^{opt}\cap S$, we can write
\[
v\left(\big(\calA_{k}^{opt}\cap S\big)\setminus\{i_0\}\right)=\sum\limits_{\substack{i\in\calA_{k}^{opt}\cap S\\i\neq i_0}}v\left(\calA_{i}^{opt}\right)\le
c\left(\big(\calA_{k}^{opt}\cap S\big)\setminus\{i_0\}\right)\frac{v\left(\calA_{k}^{(2)}\right)}{c\left(\calA_{k}^{(2)}\right)}\le v\left(\calA_{k}^{(2)}\right),
\]
where the last inequality follows from the fact
$c\left(\big(\calA_{k}^{opt}\cap S\big)\setminus\{i_0\}\right)\le
C_k-c\big(\calA_{i_0}^{opt}\big)<c\big(\calA_{k}^{(2)}\big)$. Then
we have
$v\left(\calA_{k}^{(1)}\right)+v\left(\calA_{k}^{(2)}\right)\ge
v\left(\big(\calA_{k}^{opt}\cap
S\big)\setminus\{i_0\}\right)+v\big(\calA_{i_0}^{opt}\big)=
v\big(\calA_{k}^{opt}\cap S\big)$. Hence, $\sum\limits_k
v\left(\calA_{k}^{(1)}\right)+\sum\limits_k
v\left(\calA_{k}^{(2)}\right)\ge\sum\limits_{k}v\big(\calA_{k}^{opt}\cap
S\big)$.

Therefore, we have
\begin{eqnarray*}
2v\big(\calA^{(1)}\big)+2v\big(\calA^{(2)}\big)&=& \sum\limits_i v\left(\calA_{i}^{(1)}\right)+\sum\limits_i v\left(\calA_{i}^{(2)}\right)+\sum\limits_k v\left(\calA_{k}^{(1)}\right)+\sum\limits_k v\left(\calA_{k}^{(2)}\right)\\
&\ge& \sum\limits_{i\notin S} v\left(\calA_{i}^{opt}\right)+\sum\limits_{k}v\left(\calA_{k}^{opt}\cap S\right)\\
&=&v(\calA^{opt}).
\end{eqnarray*}
This completes the proof.
\end{proof}

\subsection{Example for \multigs}\label{example-sm-alg-not-truthful}

The following example shows that the \multigs\ is not necessarily
truthful for the general \gap.

\begin{example}
There are four jobs $\{1,2,3,4\}$ and three machines $\{x,y,z\}$.
The values, capacities, and preferences of compatible pairs are
given by the following table:

\begin{table*}[h]
\small
\begin{center}
\begin{tabular}{|c|c|c|c|c|c|c|c|c|c|c|c|c|c|}\hline
 &  & \multicolumn{3}{|c|}{Job 1} & \multicolumn{3}{|c|}{Job 2} & \multicolumn{3}{|c|}{Job 3} & \multicolumn{3}{|c|}{Job 4}  \\
 \cline{2-14}
 & Values & \multicolumn{3}{|c|}{$v_{1x}=1, v_{1y}=0.5$} & \multicolumn{3}{|c|}{$v_{2x}=1, v_{2z}=0.5$} & \multicolumn{3}{|c|}{$v_{3x}=10, v_{3z}=20$} & \multicolumn{3}{|c|}{$v_{4x}=5, v_{4y}=0.1$}  \\
 \cline{2-14}
\raisebox{1.3ex}[0pt]{Jobs} & Capacities & \multicolumn{3}{|c|}{$c_{1x}=0.5, c_{1y}=1$} & \multicolumn{3}{|c|}{$c_{2x}=0.5, c_{2z}=1$} & \multicolumn{3}{|c|}{$c_{3x}=1, c_{3z}=100$} & \multicolumn{3}{|c|}{$c_{4x}=1, c_{4y}=1$} \\
 \cline{2-14}
 & Preferences & \multicolumn{3}{|c|}{$\calL_1:x > y$} & \multicolumn{3}{|c|}{$\calL_2:x > z$} & \multicolumn{3}{|c|}{$\calL_3:z > x$} & \multicolumn{3}{|c|}{$\calL_4:x > y$} \\ \hline \hline
 & & \multicolumn{4}{|c|}{Machine $x$} & \multicolumn{4}{|c|}{Machine $y$} & \multicolumn{4}{|c|}{Machine $z$}  \\
 \cline{2-14}
Machines & Capacities & \multicolumn{4}{|c|}{$C_x = 1$} & \multicolumn{4}{|c|}{$C_y = 1$} & \multicolumn{4}{|c|}{$C_z = 100$}  \\
 \cline{2-14}
 & Preferences & \multicolumn{4}{|c|}{$\calL_x: \{3\} > \{4\} > \{1=2\}$} & \multicolumn{4}{|c|}{\hspace{0.8cm} $\calL_y: \{1\} > \{4\} \hspace{0.8cm}$} & \multicolumn{4}{|c|}{$\calL_z: \{2\} > \{3\}$}  \\
\hline
\end{tabular}
\end{center}
\end{table*}
\normalsize \noindent If all agents report their true
information, jobs 1 and 2 propose to $x$, and job 3 proposes to $z$;
then job 4 proposes to $x$, which kicks jobs 1 and 2 off to $y$ and
$z$, respectively. Job 3 is kicked off by 2 then and next proposes
to $x$, which kicks job 4 off. Hence, job 4 does not get any
assignment eventually. However, if job 4 hides $x$ and only reports
machine $y$, then he will get $y$ by the stable matching algorithm.
Therefore, the stable matching algorithm is not truthful for the
general \gap.
\end{example}


\subsection{Analysis of \lsizem($\lambda$), \separationm($\lambda$), \Gapm($\lambda,\mu$)}

\begin{claim}\label{cl:large_approx}
\lsizem($\lambda$) is a truthful mechanism. Further, for any
instance where all pairs $(i,k)$ satisfy
$c_{ik}\ge\frac{1}{\lambda}C_k$, its approximation ratio is at most
$2\lambda$.
\end{claim}
\begin{proof}
For truthfulness, reporting pairs with $c_{ik}<\frac{1}{\lambda}C_k$
will not change anything in the mechanism. Therefore, we may
restrict ourselves to the setting only with pairs
$c_{ik}\ge\frac{1}{\lambda}C_k$. Then the greedy strategy to pick
job-machine pairs implies truthfulness for all jobs.

Consider the optimal assignment $\calA^{opt}$ of jobs to machines.
Note that for each machine $k$ we have $\big|\calA_{k}^{opt}\big|\le
\lambda$ (we recall that $\calA_{k}^{opt}$ is the set of jobs that
machine $k$ is assigned to). Let $\calA$ be the assignment generated
by our mechanism. Let us consider a pair $(i,k)$ in
$\calA_{k}^{opt}$ with the highest value $v_{ik}$. Let
$v\left(\calA_i\right)$ and $v(\calA^{opt}_i)$ be the values of the
assignments of job $i$ in $\calA$ and $\calA^{opt}$, respectively.
In the assignment $\calA$, there are two possibilities for job $i$:
\begin{itemize}
\item either $v(\calA_i)\ge v(\calA^{opt}_i)$
\item or $v(\calA_i)< v(\calA^{opt}_i)$. For the latter case, by the greedy rule, machine $k$ should be
assigned in $\calA$ to a job with a value greater than or equal to
$v_{ik}$.
\end{itemize}
Hence, we have
\[
2v(\calA)=\sum_{k}\sum_{i\in \calA_k}v_{ik}+\sum_{k}v(\calA_k)\ge\sum_{k}\max_{i\in \calA_k^{opt}}v_{ik}\ge\frac{1}{\lambda}v(\calA^{opt}).
\]
Therefore, $2\lambda\cdot v(\calA)\ge v(\calA^{opt})$ and the claim follows.
\end{proof}

\begin{claim}\label{cl:mechanism-2}
\separationm($\lambda$) is a truthful mechanism.
\end{claim}
\begin{proof}
Observe that reporting pairs with $c_{ik}>\frac{1}{\lambda}C_k$ will
not change anything in the mechanism. Therefore, we may restrict
ourselves to the setting only with pairs
$c_{ik}\le\frac{1}{\lambda}C_k$. Note that the value $\lambda$,
i.e., the number of jobs that each machine can take, is independent
of the bids of all jobs. Thus, the greedy strategy to pick
job-machine pairs implies truthfulness for all jobs.
\end{proof}

\begin{claim}
\Gapm($\lambda,\mu$) is a universally truthful mechanism.
\end{claim}
\begin{proof}
While \Gapm($\lambda,\mu$) is not deterministic, it does not use any
reported information from jobs to generate the testing group $T$.
Once $T$ is generated, the mechanism performs deterministically.
Therefore, it runs on a distribution over deterministic mechanisms.
Further, notice that every job in the testing group $T$
derives utility zero. Therefore, given the fixed set $T$, none of
the jobs in $T$ can benefit from reporting untruthfully. On the
other hand, for each remaining job $i\in R$, reporting the truth
preference maximizes its utility, since jobs are processed one by
one in a fixed given order and every time we search for the best
possible assignment for the processed job. (Note that the threshold
value $t_k$ is computed according to the outcome of $\calA^T$ from
the testing group $T$; thus it is independent of the bids of any job
in $R$.)
\end{proof}

\begin{claim}\label{cl:stable_3}
The assignment $\calA^T$ returned by the
mechanism \Gapm($\lambda,\mu$) is stable with respect to all jobs in
$T$ and all machines with virtual capacity
$\frac{\lambda-1}{\lambda}C_k$.
\end{claim}
\begin{proof}
Assume the contrary, that $i$ and $k$ form a blocking pair, where
$i$ strictly prefers $k$ to $\calA_i$ and $k$ can get a better
assignment from jobs in $\calA_k\cup\{i\}$ (with respect to the
virtual capacity constraint). Since every job proposes according to
its preference list $\calL_i$, we know that at a certain step job $i$
has proposed to $k$ during the execution of the algorithm. Then
either $i$ got rejected right away, or got assigned but rejected
later (due to proposals from other jobs). For both cases, from that
moment until the end of the algorithm, because the capacity of all
considered pairs is at most $\frac{1}{\lambda}C_k$, we know that (i)
the total capacity of the assigned jobs to $k$ is strictly larger
than the virtual capacity $\frac{\lambda-1}{\lambda}C_k$ (and, of
course, less than or equal to the real capacity $C_k$), and (ii) $k$
prefers every assigned job to $i$. These two facts imply that $k$
cannot get a better assignment from $\calA_k\cup\{i\}$, a
contradiction.
\end{proof}

\subsection{Approximation Analysis of \GapMainm}\label{appendix-main-mechanism-approx}

Let $OPT$ be the optimal allocation. We denote by
$\calA^{(1)}$, $\calA^{(2)}$ and $\calA^{(3)}$ the allocations
obtained in the mechanisms \lsizem$(\lambda)$,
\separationm$(\lambda)$ and \Gapm($\lambda,\mu$), respectively. The
expected value of the mechanism \GapMainm\ is therefore
$\frac{1}{3}\big(v(\calA^{(1)})+v(\calA^{(2)})+v(\calA^{(3)})\big)$.


We divide all pairs into two groups: a large group containing the
pairs $(i,k)$ with $c_{ik}> \frac{C_k}{\lambda}$, and a small group
containing the remaining pairs. Let $OPT^{small}$ and $OPT^{large}$
denote the optimal allocations for the settings restricted on the
jobs only in the small group and large group, respectively. Then we
have
\[
v(OPT)\le v(OPT^{small})+v(OPT^{large}).
\]

By Claim~\ref{cl:large_approx}, we know that the total value derived
from \lsizem$(\lambda)$ satisfies $v(\calA^{(1)})\ge
\frac{v(OPT^{large})}{2\lambda}$. It remains to handle
$OPT^{small}$. Instead of dealing with $OPT^{small}$ directly, we
consider the allocation $\calA^*$ defined according to
Step~\ref{step:2} of \Gapm($\lambda,\mu$) for {\em all} jobs with
respect to the pairs with capacities less than or equal to
$\frac{C_k}{\lambda}$. Formally, $\calA^*$ is defined as follows:
\begin{enumerate}
\item Remove all pairs $(i,k)$ with capacity $c_{ik}>\frac{C_k}{\lambda}$.
\item For each job $i$, let $\calL_i$ be his preference list over machines.
\item For each machine $k$, let $\calL_k$ be its preference list over jobs.
\item Run the \multigs\ with  virtual capacity $\frac{\lambda-1}{\lambda}C_k$ for each machine $k$.
\item Denote the generated assignment by $\calA^*$.
\end{enumerate}

Similar to the proof of Claim~\ref{cl:stable_3}, $\calA^*$ is a stable assignment with respect to virtual capacity $\frac{\lambda-1}{\lambda}C_k$.
The following claim implies that $v(\calA^*)$ is a constant approximation to $v(OPT^{small})$.

\begin{claim}
$(2+\frac{1}{\lambda-1})\cdot v(\calA^*)\ge v(OPT^{small})$.
\end{claim}
\begin{proof}
To simplify the notation, let $\calA^{opt}$ denote the
assignment $OPT^{small}$. Since $\calA^*$ is a stable assignment, for each
pair $(i,k)\in \calA^{opt}\setminus \calA^*$, we have either (i)
$v(\calA^{opt}_i)=v_{ik}\le v\left(\calA^*_i\right)$, or (ii)
$\frac{v_{ik}}{c_{ik}}\le \frac{v_{i'k}}{c_{i'k}}$, for every $i'\in
\calA^*_{k}$, and $c(\calA^*_k)\ge C_k\frac{\lambda-1}{\lambda}$.
The latter implies that
$\frac{v_{ik}}{c_{ik}}\le\frac{v(\calA^*_k)}{c(\calA^*_k)}\le\frac{v(\calA^*_{k})}{C_k}\frac{\lambda}{\lambda-1}$.

In the assignment $\calA^{opt}$, we denote by $X$ the set of jobs
$i$ such that $v(\calA^{opt}_i)\le v(\calA^*_i)$ and by $Y$ the
remaining jobs. Then we have
\begin{eqnarray*}
v(\calA^{opt}) &=& \sum_{i\in X}v(\calA^{opt}_i)+\sum_{k}v(\calA^{opt}_{k}\cap Y) \\
&\le& \sum_{i\in X}v(\calA^*_i)+\frac{\lambda}{\lambda-1}\sum_{k}\frac{v(\calA^*_{k})}{C_k}\sum_{i\in \calA^{opt}_{k}\cap\ Y}c_{ik}\\
&\le& v(\calA^*)+\frac{\lambda}{\lambda-1}\sum_{k}v(\calA^*_{k}) \\
&=& \left(1+\frac{\lambda}{\lambda-1}\right)v(\calA^*).
\end{eqnarray*}
Hence, we get $v(\calA^{opt})\le(2+\frac{1}{\lambda-1})v(\calA^*)$.
\end{proof}

In addition, by stability $\calA^*$ enjoys the following useful properties.

\begin{claim}\label{cl:not_big_improve} For each job $i\in T$ and machine $k$, we have
\begin{enumerate}
\item $v(\calA^{T}_i)\ge v(\calA^*_i).$
\item $v(\calA_{k}^{T})\le\frac{\lambda}{\lambda-1}v(\calA^*_{k}).$
\end{enumerate}
\end{claim}
\begin{proof}
Our proof of the claim exploits similar ideas from the standard
one-to-one stable matching. Note that both assignments $\calA^*$ and
$\calA^{T}$ are generated by the same algorithm, but $\calA^*$ is on
a bigger set of jobs. Similar to the proof of
Claim~\ref{cl:stable_3}, the virtual capacity
$\frac{\lambda-1}{\lambda}C_k$ and the fact that all considered
pairs have capacities $c_{ik}\le \frac{C_k}{\lambda}$ ensure that at
any moment of the algorithm, the set of assigned jobs to any machine
is simply the best possible subset selected from the set of all jobs
that have ever proposed to the machine up to that moment. In
particular, at the end of the algorithm, every machine $k$ will have
its most preferred feasible set (given the virtual capacity
constraint) chosen from the set of all proposals that $k$ has ever
received.

For each job $i$, denote by $best(i)$ the best possible assignment
for $i$ taken over all stable assignments. We will argue that the
\multigs\ will assign every job $i$ to $best(i)$. Assume otherwise,
since every job makes proposals by decreasing order of $\calL_i$, in
the course of the algorithm, there must be a job $i$ that has been
rejected by $best(i)$. We consider the first moment in the algorithm
when a job $i$ is rejected from the machine $k=best(i)$. At that
moment, let $S$ be the set of jobs that are assigned to $k$. By the
above argument, we know that for any $i'\in S$, $k$ prefers $i'$ to
$i$. Next consider a stable assignment $\calA'$, where $i$ is
assigned to $k$ (by the definition of $best(i)$, such assignment
exists). Then there exists $i'\in S$ that is not assigned to $k$ in
$\calA'$; let $\calA_{i'}'=k'(\neq k)$. By the definition of
$best(i')$, we know that $i'$ weakly prefers $best(i')$ to $k'$.
Since $\calA'$ is a stable assignment, $i'$ and $k$ do not form a
blocking pair, which implies that $i'$ prefers $k'$ to $k$. Now we
recall that $i$ is the first job rejected from the machine $best(i)$
in \multigs; we know that at that moment, $i'$ has not been rejected
by $best(i')$. Hence, $i'$ weakly prefers $k$ to $best(i')$.
Therefore, on the preference $\calL_{i'}$, we have the following
preference order: $k\succeq best(i')\succeq k'\succ k$, which gives
a contradiction. Therefore, each job $i$ is assigned to $best(i)$ in
the output of \multigs.

Note that the above argument does not rely on any specific order of
proposals (especially, the one defined in Step~1 of \multigs). That
is, an arbitrary order of proposals, as long as every job proposes
to the most preferred machine that he has not proposed yet, will
lead to the same assignment by matching every job $i$ to $best(i)$.
Now given that the order in which jobs propose does not matter for
the outcome of the algorithm, in the computation of $\calA^*$, we
may consider the order where jobs in $T$ are settled first and then
add the remaining jobs. Thus, for each job $i\in T$ the assignment
$\calA^*_i$ can only be worse than in $\calA^T_i$, i.e.,
$v(\calA^{T}_i)\ge v(\calA^*_i).$

In addition, we know that the set of jobs proposed to each machine
$k$ in $\calA^{T}$ is a subset in $\calA^*$. Thus, by the above argument,
the average ratio $\frac{v_{ik}}{c_{ik}}$ of the assigned jobs in $\calA^*_{k}$ is larger than or equal to that in $\calA^{T}_{k}$.
Notice that in the assignment $\calA^T_k$, the total capacity used is at most $C_k$.
For the assignment $\calA^*_k$, either we use more than $(1-\frac{1}{\lambda})C_k$ capacity or
otherwise (which implies that $\calA^{T}_k\subseteq \calA^*_{k}$). For either case, we have
$v(\calA_{k}^{T})\le\frac{\lambda}{\lambda-1}v(\calA^*_{k})$, which completes the proof.
\end{proof}

For the assignment $\calA^*$, let us consider a restricted
assignment $\widetilde{\calA}^{*}$, where we only keep the first
$\lambda$ largest value jobs of $\calA^*$ for each machine $k$. Note
that $\widetilde{\calA}^{*}$ is a feasible assignment for the setting in
\separationm\ as well. We have the following claim.

\begin{claim}\label{claim-a*-a2}
$v\big(\calA^{(2)}\big)\ge\frac{1}{2}v\big(\widetilde{\calA}^{*}\big).$
\end{claim}
\begin{proof}
Observe that both $\calA^{(2)}$ and $\widetilde{\calA}^{*}$ run on the same set pairs, i.e., $c_{ik}\le \frac{C_k}{\lambda}$.
We may split each machine into $\lambda$ equal slots, where each slot can take only one job. Thus, $\widetilde{\calA}^{*}$
may be viewed as a matching and $\calA^{(2)}$ as a maximal matching for the new setting between jobs and slots.
Since a maximal matching is a 2 approximation to the maximum matching, which is an upper bound on $v\big(\widetilde{\calA}^{*}\big)$,
we have $2\cdot v\big(\calA^{(2)}\big)\ge v\big(\widetilde{\calA}^{*}\big)$.
\end{proof}

Before continuing the proof, we first establish the following probabilistic fact.

\begin{lemma}\label{lem:probability}
Let $\delta_1\le\frac{1}{36}$ be a positive real number. Let $a_1\ge
a_2\ge \cdots\ge a_\ell\ge 0$ be real numbers with the sum
$a=a_1+a_2+\ldots+a_\ell$. Assume $a_1 < \delta_1 a$. Consider the
following probabilistic event: One selects each number
$a_1,\ldots,a_\ell$ independently at random with probability $1/2$
each; let $b$ be the random variable denoting the summation of the
selected numbers. Then
\[
\pr\left(\frac{1}{3}a<b<\frac{2}{3}a\right)\ge \frac{3}{4}.
\]
\end{lemma}
\begin{proof}
Let us consider for each $1\le j\le\ell$, the random variables $X_j$
with $\pr(X_j=0)=\pr(X_j=a_j)=0.5$. Let
$X=\sum\limits_{i=j}^{\ell}X_j$; we have $b=X$. Then the expectation
$\ex(X)=\frac{a}{2}$ and variance
\[
\sigma^2=\mathbf{Var}(X)=\sum\limits_{j=1}^{\ell}\mathbf{Var}(X_j)=\frac{1}{4}\sum\limits_{j=1}^{\ell}a_j^2.
\]
Applying Chebyshev's inequality, we get
\[
\pr\left(\left|X-\frac{a}{2}\right|\ge 2\sigma\right)\le\frac{1}{4}.
\]

In order to conclude the proof of the lemma, it remains to show that
$2\sigma\le \frac{a}{6}$, which is equivalent to show that
\[
36\cdot\left(a_1^2+a_2^2+\ldots+a_\ell^2\right)\le\left(a_1+a_2+\ldots+a_\ell\right)^2.
\]
Since $a\ge \frac{a_1}{\delta_1}\ge\frac{a_j}{\delta_1}$ for every $1\le j\le\ell$, we have
\[
\left(a_1+a_2+\ldots+a_\ell\right)^2=\sum_{j=1}^{\ell}a_j\cdot a\ge\sum_{j=1}^{\ell}\frac{a_j^2}{\delta_1} \ge 36\sum_{j=1}^{\ell}a_j^2 .
\]
Therefore, the lemma follows.
\end{proof}

Consider the partition of jobs into the two sets $T$ and $R$ in the
\Gapm($\lambda,\mu$). For a fixed machine $k$, we consider the jobs
$\calA_{k}^*\cap T$ and $\calA_{k}^*\cap R$. By the definition of
$T$ and $R$, every job in $\calA_{k}^*$ will be placed in $T$ with
probability $1/2$. Recall that $\widetilde{\calA}^{*}_{k}$ keeps the
top $\lambda$ value jobs in $\calA^*_k$. Depending on the relation
between $v(\widetilde{\calA}^{*}_{k})$ and $v(\calA^*_{k})$, at
least one of the following alternatives has to be true (where the
latter follows from the above Lemma~\ref{lem:probability}):
\begin{enumerate}
\item $v(\widetilde{\calA}^{*}_{k})\ge \delta_1 v(\calA^*_{k}),$
\item $\pr_{_T}\Big(\frac{1}{3}v(\calA^*_{k})\le v(\calA^*_{k}\cap T) \le\frac{2}{3}v(\calA^*_{k})\Big)\ge\frac{3}{4}.$
\end{enumerate}


Intuitively, if there are many similar jobs in $\calA_{k}^*$, then
with a high probability the total value of jobs in $\calA_{k}^*\cap
T$ and in $\calA_{k}^*\cap R$ will be close to each other. On the
other hand, if the former does not hold, the \separationm\ gives us
a good value for machine $k$. Let us denote the set of the
machines satisfying the condition $v(\widetilde{\calA}^{*}_{k})\ge
\delta_1 v(\calA^*_{k})$ by $G$ and the remaining set of machines by
$H$. Thus, by Claim~\ref{claim-a*-a2}, \separationm\ guarantees that
we get a constant fraction of the total value of $\calA^*$ taken
over all machines in $G$. For every machine $k\in H$, Lemma
\ref{lem:probability} ensures that
\[
\frac{1}{3}v(\calA^*_{k})\le v(\calA^*_{k}\cap
T)\le\frac{2}{3}v(\calA^*_{k})
\]
occurs with a probability of at least $3/4$. Consider the collection
$\calD^*_k$ of realizations of $T$ for which the above two
inequalities hold; by
abusing the notation, let $\calD^*_k$ denote the distribution of $T$
restricted to the collection as well.

Now let us estimate the value we get from \Gapm($\lambda,\mu$) over
$H$. For each machine $k\in H$, there are the following two cases.
\begin{enumerate}
\item In \Gapm($\lambda,\mu$), we have rejected a job at machine $k$ due to the capacity
constraint, that is, we have rejected a pair $(i,k)$ because of
$c_{ik}+c(\calA_k)> C_k$. Since \Gapm($\lambda,\mu$) only considers
pairs with capacities less than or equal to $\frac{C_k}{\lambda}$,
in that case almost all the capacity of $k$ is used. That is,
\[
\sum_{i\in\calA^{(3)}_{k}}c_{ik}\ge C_k\left(1-\frac{1}{\lambda}\right).
\]
Further, by the rule of the mechanism, for every $i\in
\calA_{k}^{(3)}$, we have $\frac{v_{ik}}{c_{ik}}\ge t_k$, i.e.,
$v_{ik}\ge c_{ik}\frac{\mu\cdot v(\calA_{k}^{T})}{C_k}$. Let
$\delta_2=\mu(1-\frac{1}{\lambda})>0$; therefore, we have
\begin{equation}
\label{eq:case1}
v(\calA_{k}^{(3)})\ge \sum_{i\in\calA^{(3)}_{k}}c_{ik} \frac{\mu\cdot v(\calA_{k}^{T})}{C_k} \ge \mu\left(1-\frac{1}{\lambda}\right)v(\calA_{k}^{T})=\delta_2\cdot v(\calA_{k}^{T}).
\end{equation}

\item No job who has passed the threshold $t_k$ has been rejected from
machine $k$ in \Gapm($\lambda,\mu$). Denote by $R^{+}$ the set of
jobs in $R$ who have passed the threshold for the corresponding
machine in assignment $\calA^*$; denote by $R^{-}$ the remaining
jobs in $R$. Therefore, all jobs in $\calA^*_{k}\cap R^{+}$ get at
least as good an assignment as in $\calA^*$. Then for any
$i\in\calA^*_{k}\cap R^+$ we have $v(\calA^{(3)}_i)\ge
v(\calA^*_i).$

Further, for each job $i\in \calA_{k}^*\cap R^-$, we have
\[
v_{ik}<c_{ik}t_k=\mu\frac{v(\calA^{T}_{k})c_{ik}}{C_k}.
\]
Therefore,
\[v(\calA^*_{k}\cap R^-)\le \mu\cdot v(\calA^{T}_{k}) \frac{\sum_{i\in\calA^*_{k}\cap R^-}c_{ik}}{C_k}\le \mu\cdot v(\calA^{T}_{k})\le \mu\frac{\lambda}{\lambda-1}v(\calA^*_{k})\]
where the last inequality follows from
Claim~\ref{cl:not_big_improve}. Thus,
\[\sum\limits_{i\in\calA^*_{k}\cap R}v(\calA^{(3)}_i)\ge \sum\limits_{i\in\calA^*_{k}\cap R^+}v(\calA^{(3)}_i)\ge v(\calA^*_{k}\cap R) - \mu\frac{\lambda}{\lambda-1}v(\calA^*_{k}).\]
Let $\delta_3=\frac{1}{3}- \mu\frac{\lambda}{\lambda-1}$. Taking appropriate values for $\lambda$ and $\mu$, we can ensure that $\delta_3>0$. For any $T\in \calD^*_k$, we have $v(\calA^*_{k}\cap R)\ge\frac{1}{3}v(\calA^*_{k})$; then we get
\begin{equation}
\label{eq:case2}
\sum\limits_{i\in\calA^*_{k}\cap R}v(\calA^{(3)}_i)\ge \left(\frac{1}{3}- \mu\frac{\lambda}{\lambda-1}\right)v(\calA^*_{k}) = \delta_3\cdot v(\calA^*_{k}).
\end{equation}
\end{enumerate}

Next we estimate the expectation of $v(\calA^{(3)})$. Let $\calD$
denote the distribution of the mechanism \Gapm($\lambda,\mu$)\ to
generate $T$.
\begin{eqnarray*}
\ex_{T\sim \calD}\bigg(2\cdot v\left(\calA^{(3)}\right)\bigg) &=& \ex_{T\sim \calD}\Bigg(\sum_{k}v\left(\calA^{(3)}_{k}\right)+\sum_{k}\sum_{i\in\calA_{k}^*\cap R}v\left(\calA^{(3)}_i\right)\Bigg) \\
&\ge& \sum_{k\in H}\ex_{T\sim \calD}\Bigg(v\left(\calA^{(3)}_{k}\right)+\sum_{i\in\calA_{k}^*\cap R}v\left(\calA^{(3)}_i\right)\Bigg) \\
&\ge& \sum_{k\in H}\pr(T\in \calD^*_k)\cdot \ex_{T\sim \calD^*_k}\Bigg(v\left(\calA^{(3)}_{k}\right)+\sum_{i\in\calA_{k}^*\cap R}v\left(\calA^{(3)}_i\right)\Bigg).\\
\end{eqnarray*}
The last inequality follows from the fact that for any random
variable $\xi(t)$, one has
\[\ex_{x}\Big(\xi(x)\Big)=\pr(x\mid Z)\cdot \ex_{x:\,x\in Z}\Big(\xi(x)\Big)+\pr\left(x\mid \overline{Z}\right)\cdot \ex_{x:\,x\in \overline{Z}}\Big(\xi(x)\Big).\]
We continue the argument; by applying either \eqref{eq:case1} or
\eqref{eq:case2}, we have
\begin{eqnarray*}
\ex_{T\sim \calD}\left(2\cdot v\left(\calA^{(3)}\right)\right) &\ge& \sum_{k\in H}\pr(T\in \calD^*_k)\cdot \ex_{T\sim \calD^*_k}\min\Big(\delta_2\cdot v\left(\calA^{T}_{k}\right), \delta_3\cdot v\left(\calA_{k}^*\right)\Big) \\
&\ge& \min\Big(\delta_2,\frac{(\lambda-1)\delta_3}{\lambda}\Big)\cdot\sum_{k\in H}\pr(T\in \calD^*_k)\cdot \ex_{T\sim \calD^*_k} \Big(v\left(\calA^{T}_{k}\right)\Big)
\end{eqnarray*}
where the last inequality follows from
Claim~\ref{cl:not_big_improve}. Let
$\delta_4=\min\left(\delta_2,\frac{(\lambda-1)\delta_3}{\lambda}\right)>0$.

Recall that
\[
\ex_{T\sim \calD}\Big(v(\calA^{T}_{k})\Big)= \pr(T\in \calD^*_k)\cdot\ex_{T\sim \calD^*_k}\Big(v(\calA^{T}_{k})\Big)+ \pr(T\notin \calD^*_k)\cdot\ex_{T\sim \calD\setminus\calD^*_k}\Big(v(\calA^{T}_{k})\Big).
\]
Since $\pr(T\notin \calD^*_k)\le \frac{1}{4}$, by Claim~\ref{cl:not_big_improve}, which says that $v(\calA_{k}^{T})\le\frac{\lambda}{\lambda-1}v(\calA^*_{k})$, for each $k\in H$ we get
\begin{equation}
\label{eq:cond_expect}
\pr(T\in \calD^*_k)\cdot\ex_{T\sim \calD^*_k}\Big(v(\calA^{T}_{k})\Big)\ge\ex_{T\sim
\calD}\Big(v(\calA^{T}_{k})\Big)-\frac{1}{4}\frac{\lambda}{\lambda-1}v(\calA_{k}^*).
\end{equation}

We continue our lower bound on $\ex_{T\sim \calD}\Big(2\cdot
v\left(\calA^{(3)}\right)\Big)$; by applying \eqref{eq:cond_expect},
we have
\begin{eqnarray*}
\ex_{T\sim \calD}\Big(2\cdot v\left(\calA^{(3)}\right)\Big) &\ge& \delta_4\cdot\sum_{k\in H}\left(\ex_{T\sim \calD}\Big(v\left(\calA^{T}_{k}\right)\Big)-\frac{\lambda}{4(\lambda-1)}v\left(\calA_{k}^*\right)\right) \\
&\ge& \delta_4\cdot\sum_{k\in H}\ex_{T\sim \calD}\Big(v\left(\calA^{T}_{k}\right)\Big)-\delta_4\cdot\frac{\lambda}{4(\lambda-1)}v\left(\calA^*\right) \\
&\ge& \delta_4\cdot\sum_{k}\ex_{T\sim \calD}\Big(v\left(\calA^{T}_{k}\right)\Big)-\delta_4\cdot\frac{\lambda}{\lambda-1}\sum_{k\in G}v\left(\calA^*_{k}\right)-\delta_4\cdot\frac{\lambda}{4(\lambda-1)}v\left(\calA^*\right).\\
\end{eqnarray*}

By using the first property of Claim \ref{cl:not_big_improve}, we
have
\begin{eqnarray*}
\sum_{k}\ex_{T\sim \calD}\Big(v(\calA^{T}_{k})\Big) &=& \ex_{T\sim \calD}\Big(v\left(\calA^T\right)\Big)=\ex_{T\sim \calD}\sum_{i\in T}v\left(\calA^T_i\right) \\
&\ge& \ex_{T\sim \calD}\sum_{i\in T}v\left(\calA^*_i\right)
=\frac{1}{2}\cdot v\left(\calA^*\right).\\
\end{eqnarray*}

Then we have the following lower bound:
\begin{equation}
\label{eq:expect_A_3}
\ex_{T\sim \calD}\Big(v\big(\calA^{(3)}\big)\Big)\ge \delta_4\cdot\left(\frac{1}{4}-\frac{\lambda}{8(\lambda-1)}\right)v\left(\calA^*\right)-\frac{\delta_4\lambda}{2(\lambda-1)}\cdot \sum_{k\in G}v\left(\calA^*_{k}\right).
\end{equation}
Let
$\delta_5=\delta_4\cdot\left(\frac{1}{4}-\frac{\lambda}{8(\lambda-1)}\right)$.
Taking an appropriate value for $\lambda$, we can ensure that
$\delta_4,\delta_5>0$.

Recall the definition of group $G$ and what we already got for
$v\left(\calA^{(2)}\right)$ by Claim~\ref{claim-a*-a2}:
\[
v\left(\calA^{(2)}\right)\ge\frac{1}{2}\cdot v\left(\widetilde{\calA}^{*}\right)\ge\frac{\delta_1}{2}\cdot \sum_{k\in G}v\left(\calA^*_{k}\right).
\]
Taking appropriate values for $\delta_1$ and $\delta_4$ (such that
$\frac{\delta_1}{2}\ge\frac{\delta_4\lambda}{2(\lambda-1)}$), we
have
\[
\ex_{T\sim \calD}\Big(v\big(\calA^{(3)}\big)\Big)+v\big(\calA^{(2)}\big)\ge \delta_5\cdot v\left(\calA^*\right).
\]

Finally, the expected value of the \GapMainm($\lambda,\mu$)
satisfies
\begin{eqnarray*}
& & \frac{1}{3}\cdot \left(v\big(\calA^{(1)}\big)+v\big(\calA^{(2)}\big)+\ex_{T\sim \calD}\Big(v\big(\calA^{(3)}\big)\Big)\right) \\
&\ge& \frac{1}{3}\cdot \left(\frac{1}{2\lambda}v\left(OPT^{large}\right)+\frac{\delta_5}{2+\frac{1}{\lambda-1}}\cdot v\left(OPT^{small}\right)\right) \\
&\ge& \frac{1}{3}\cdot \min\left(\frac{1}{2\lambda},\frac{\delta_5}{2+\frac{1}{\lambda-1}}\right)\cdot v\left(OPT\right).
\end{eqnarray*}
By choosing, e.g., $\lambda=3$, $\mu=\frac{1}{6}$ and
$\delta_1=\frac{1}{36}$, it can be seen that all conditions for
$\delta_2,\delta_3,\delta_4,\delta_5$ are satisfied. Hence,
\GapMainm\ gives a constant approximation. This completes the proof
of Theorem~\ref{thm-gap-main-mechanism}.

\end{document}